\documentclass[twocolumn]{autart}    

\usepackage{graphicx}
\usepackage{setspace}
\usepackage{mathtools}
\usepackage{graphics} 
\usepackage{epsfig} 
\usepackage{mathptmx} 
\usepackage{times} 
\usepackage{amsmath} 
\usepackage{amssymb}  
\usepackage{floatrow}
\usepackage{enumerate}
\usepackage{mathrsfs}
\usepackage{amsfonts}
\usepackage{breakcites}
\usepackage{amsmath,amsfonts}
\usepackage{algorithmic}

\newcommand{\mi}{\mathrm{i}}
\usepackage[utf8]{inputenc}
\inputencoding{latin1}
% ...
\inputencoding{utf8}

\newenvironment{proof}{\par{\itshape Proof}.\ }

\newtheorem{theorem}{\indent Theorem}
\newtheorem{lemma}[theorem]{\indent Lemma}

\newtheorem{proposition}[theorem]{\indent Proposition}
\newtheorem{definition}[theorem]{\indent Definition}

\newtheorem{problem}{Problem}
\allowdisplaybreaks

\begin{document}

\begin{frontmatter}
\title{On the capability of a class of quantum sensors\thanksref{footnoteinfo}} %

\thanks[footnoteinfo]{This work was supported by the Australian Research Council's Discovery Projects funding scheme under Projects DP190101566 and DP180101805, the U.S. Office of Naval Research Global under Grant N62909-19-1-2129 and the Air Force Office of Scientific Research and the Office of Naval Research Grants under agreement number FA2386-16-1-4065.}
\author[ADFA,ANU]{Qi Yu}\ead{yuqivicky92@gmail.com},
\author[ADFA,Griffith]{Yuanlong Wang}\ead{yuanlong.wang.qc@gmail.com},
\author[ADFA]{Daoyi Dong}\ead{daoyidong@gmail.com},
\author[ANU]{Ian R. Petersen}\ead{i.r.petersen@gmail.com}

\address[ADFA]{School of Engineering and Information Technology, University of New South Wales, Canberra, ACT 2600, Australia}
\address[ANU]{Research School of Electrical, Energy and Materials Engineering, Australian National University, Canberra, ACT 2601, Australia}
\address[Griffith]{Centre for Quantum Dynamics, Griffith University, Brisbane, QLD 4111, Australia}

\begin{keyword}                           
Quantum sensor, spin chain system, similarity transformation approach, quantum sensing, quantum system identification.
\end{keyword}                             

\begin{abstract}                          
Quantum sensors may provide extremely high sensitivity and precision to extract key information in a quantum or classical physical system. A fundamental question is whether a quantum sensor is capable of uniquely inferring unknown parameters in a system for a given structure of the quantum sensor and admissible measurement on the sensor. In this paper, we investigate the capability of a class of quantum sensors which consist of either a single qubit or two qubits. A quantum sensor is coupled to a spin chain system to extract information of unknown parameters in the system. With given initialisation and measurement schemes, we employ the similarity transformation approach and the Gr$\ddot{\text{o}}$bner basis method to prove that a single-qubit quantum sensor cannot effectively estimate the unknown parameters in the spin chain system while the two-qubit quantum sensor can. The work demonstrates that it is a feasible method to enhance the capability of quantum sensors by increasing the number of qubits in the quantum sensors for some practical applications.
\end{abstract}

\end{frontmatter}

\section{Introduction}
The development of emerging quantum technologies has shown powerful potential in many applications of estimating, identifying, simulating and controlling microscale physical and chemical systems \cite{sensing2017Degen,shi2016eaching,zorzi2014minimum,Yonezawa2012quantum,hou201614qubit,Di2008}. Among quantum technologies, quantum sensing technology has been recognized as a critical advanced technology to achieve extremely high sensitivity and precision in measuring a physical (quantum or classical) quantity using a quantum apparatus. Several quantum systems such as spin systems and trapped ions have been used to develop quantum sensors for detecting magnetic fields, electric fields or temperature \cite{Bonato2017,sensing2017Degen,rotation2017Campbell,Pog2018,sensing2013shi,qibo2017adaptive}. Quantum sensors will also have wide applications in quantum system identification \cite{junzhang2014,Burgarth2017Hamiltonian,Bonnabel2009Observer,yuanlong2019Automatica,Yonezawa2012quantum,xiang2011entangle,Levitt2018Power,Shu2016Hamiltonian,Fu2016Hamiltonian,hou201614qubit}, quantum control \cite{Pog2018,zhang2017feedback,guofeng2012survey,Dong2019learning,coherent1999Nakamura,Dong2010survey,yuguo2019vanishing,chuncun2020Attosecond,Wiseman2010book} and quantum network design \cite{Yuzuru2014,Chris2005}. A lot of effort has been devoted to the design of new quantum sensors and the enhancement of the sensitivity and precision of quantum sensors \cite{rotation2017Campbell,sensing2017Degen,Pog2018,sensing2013shi,qiyu2020IFAC}. This paper focuses on analysing the capability of a class of quantum sensors.

In this paper, we define the capability of a quantum sensor as whether the data obtained by the sensor can estimate all of the unknown parameters one is interested in. In particular, we consider that the quantum sensor consists of either a single qubit (single qubit sensor) or a two-qubit system (two qubit sensor). Qubit systems are a class of fundamentally important systems with wide applications, especially in quantum information processing. The capability of qubit sensors is related to the structure of the sensors and the measurement schemes on the sensors. As an illustration, we consider the object system as a spin-$\frac{1}{2}$ chain system with given Hamiltonian structure. The objective is to estimate unknown parameters in the Hamiltonian by coupling a quantum sensor to the system and making measurement on the sensor. With given measurement and initialisation schemes, we prove that a single-qubit quantum sensor cannot effectively estimate the unknown parameters in the spin chain system while the two-qubit quantum sensor can. The results demonstrate that it is a feasible method to enhance the capability of quantum sensors by increasing the number of qubits in the quantum sensors in some practical applications. It is also worth mentioning that our work has close connection with several existing results on Hamiltonian identifiability and Hamiltonian identification. For example, the Hamiltonian identification problem has been investigated in \cite{junzhang2014,Leghtas2011Hamiltonian,Burgarth2012identification,Burgarth2009indirect,Burgarth2009Coupling,Bris2007Hamiltonian,Franco2009Hamiltonian,Burgarth2017Hamiltonian,Bonnabel2009Observer,yuanlong2019Automatica,xiang2011entangle,Levitt2018Power,Fu2016Hamiltonian}. Ref. \cite{yuanlong2020,Akira} consider the Hamiltonian identifiability problem.

The remainder of this paper proceeds as follows. In Section \ref{section2}, we give a formulation of the sensing problem and briefly introduce qubit sensors and the object system. The task of determining the capability of a sensor is summarized into two problems under two different measurement schemes. Then the problems are answered in Section \ref{section3} and Section \ref{section4}, respectively. Conclusions are presented in Section \ref{section5}.
\section{Preliminaries and problem formulation}\label{section2}
\subsection{Qubit sensor}

Qubit systems may serve as excellent sensors for information detection \cite{Akira,yuanlong2020,Pog2018}. A qubit can be a spin-$\frac{1}{2}$ system, a two-level atom or a particle in a double-well potential. In particular, we study a class of qubit sensors where several selected qubits are coupled to the object system in a specially designed way to ensure the dynamics of the sensor and the object system interact with each other. The information of interest of the object system can be acquired by probing (i.e., measuring) the sensor.


 A two dimensional complex valued vector can describe the state of a qubit, which can be further expressed as the following linear combination
 \begin{equation}
| \psi \rangle=\alpha |0\rangle + \beta |1\rangle.
 \end{equation}
 Here, $\alpha$ and $\beta$ are complex numbers satisfying the relationship $|\alpha|^2+|\beta|^2=1$,
and we may denote the $|0\rangle$ and $|1\rangle$ states as
\begin{equation}
|0\rangle=( 1\  0 )^T,  \quad
|1\rangle=( 0\  1 )^T.
\end{equation}
The following Pauli matrices
\begin{equation}
\sigma_x =\left(
\begin{array}{cc}
0 & 1\\
1 & 0\\
\end{array}\right),   \quad
\sigma_y =\left(
\begin{array}{cc}
0 & -\mi\\
\mi & 0\\
\end{array}\right),  \quad
\sigma_z =\left(
\begin{array}{cc}
1 & 0\\
0 & -1\\
\end{array}\right),
\end{equation}
together with the identity matrix $I$, form a complete basis for the observable space of a qubit. Later on, we also write $X\coloneqq \sigma_x$, $Y\coloneqq \sigma_y$ and $Z\coloneqq\sigma_z$.

The Pauli operators are non-commuting which means that Pauli measurements are not compatible with each other. The eigenstates of the Pauli matrices, corresponding to their eigenvalues $\pm 1$, are as follows
\begin{align}
  &|\psi_{x+}\rangle = \frac{1}{\sqrt{2}} (1\ 1)^T, & &|\psi_{x-} \rangle = \frac{1}{\sqrt{2}}(1 \ -1)^T, \\
  &|\psi_{y+} \rangle = \frac{1}{\sqrt{2}} (1 \ \mi)^T, & &|\psi_{y-} \rangle = \frac{1}{\sqrt{2}} (1 \ -\mi)^T, \\
  &|\psi_{z+} \rangle = (1\ 0)^T,                  & &|\psi_{z-}\rangle  = (0 \ 1)^T.
\end{align}
We assume that the initial state of every single qubit for the Q-sensor can only be prepared in an eigenstate of $\sigma_x$ (e.g., $|\psi_{x+} \rangle$). 

 The sensor coupling with the object system should be appropriately designed to effectively detect the information of interest. In this paper, we couple the sensor to the object system in the form of a string (See Fig. \ref{1-spin sensor} and Fig. \ref{2-spins sensor}) \cite{Akira,yuanlong2020}.

 The measurement capability of the sensing qubits is also essential. Specially, we assume that one can implement the measurement of $\sigma_y$ or $\sigma_z$ but not the measurement of $\sigma_x$. For example, the measurement for a single-qubit sensor can be either $\sigma_y$ or $\sigma_z$; the measurement on a two-qubit sensor can be $I\otimes\sigma_y$, $I\otimes\sigma_z$, $\sigma_y\otimes\sigma_y$, $\sigma_y\otimes\sigma_z$, $\sigma_z\otimes\sigma_y$ or $\sigma_z\otimes\sigma_z$. The readout of the sensor provides us with information on the object system. These assumptions and constraints may come from practical engineering applications of Q-sensors. For example, the manufacturer may calibrate the initial state of a Q-sensor in a specific state that can be reset and allow users to make specific measurements. However, our analysis can also be extended to other cases with different assumptions.


\subsection{Spin chain system}
\label{2.1spinsystem}

In this section, we provide a general description for the dynamics of a spin-$\frac{1}{2}$ chain system. Identifying the unknown coupling strength in a spin system is one of the basic tasks to study and control a range of quantum phenomena \cite{Jurcevic2014manybody,Paola2011,Chris2005}. The chain system is usually regarded as a black box, which means it can not be directly manipulated or measured (see the right part of the box in Fig. \ref{1-spin sensor} and Fig. \ref{2-spins sensor}), even if some limited prior knowledge on the structure of the object system can be obtained. Therefore, the initial state of the object system is assumed to be at the maximally mixed state. The spins in the object chain system only interact with their adjacent spins. Hence, we assume that the only practical way to access a chain system is through a quantum sensor, which can be both initialized and measured.

  \begin{figure}	
  	\centering		
  	\includegraphics[width=8cm]{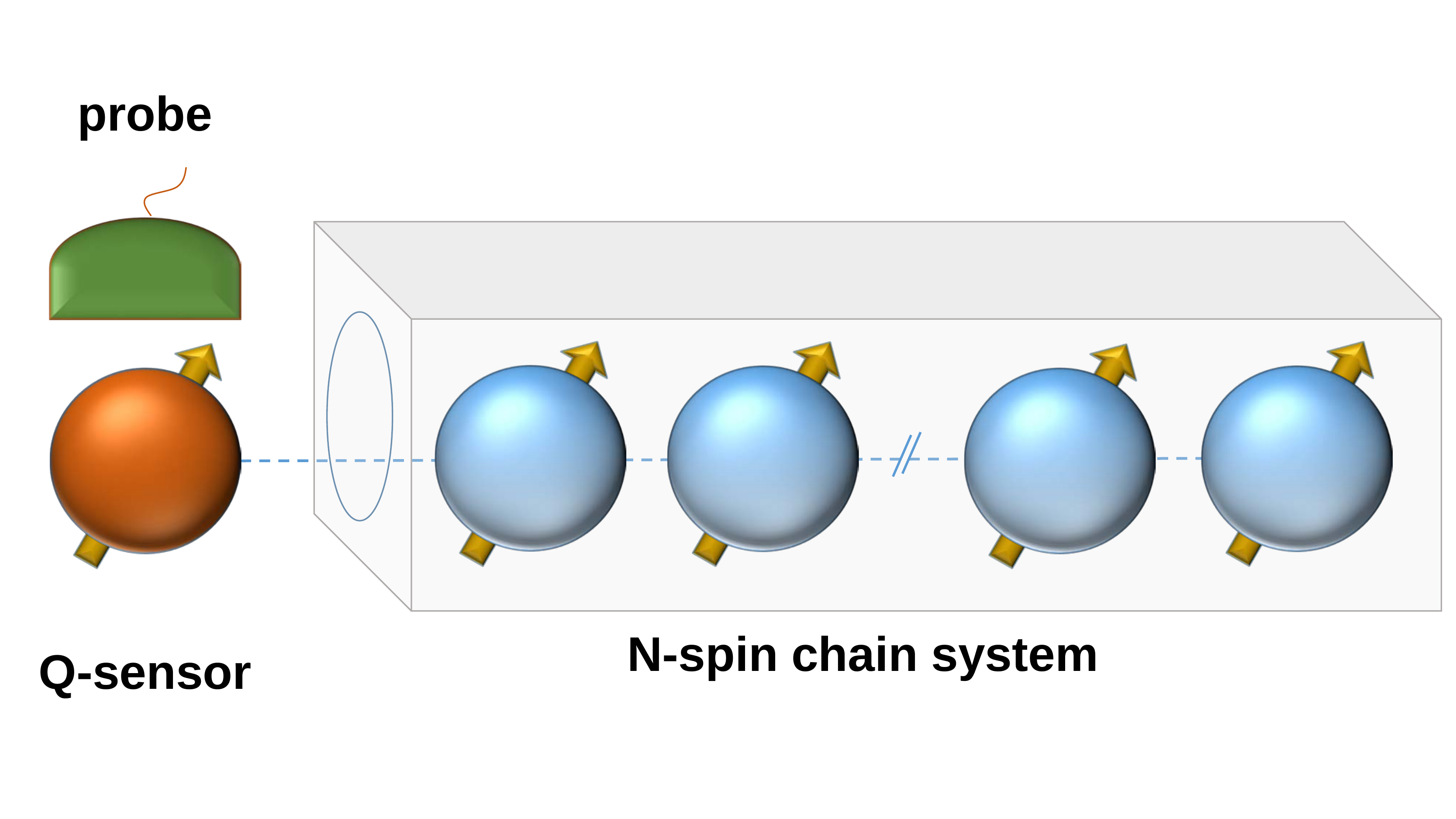}		
  	\caption{Schematic of a spin chain system connected a single-qubit sensor. The Q-sensor is measured by a probe. The object system consists of $N$ spins. }
  	\label{1-spin sensor}		
  \end{figure}
  \begin{figure}	
  	\centering		
  	\includegraphics[width=8cm]{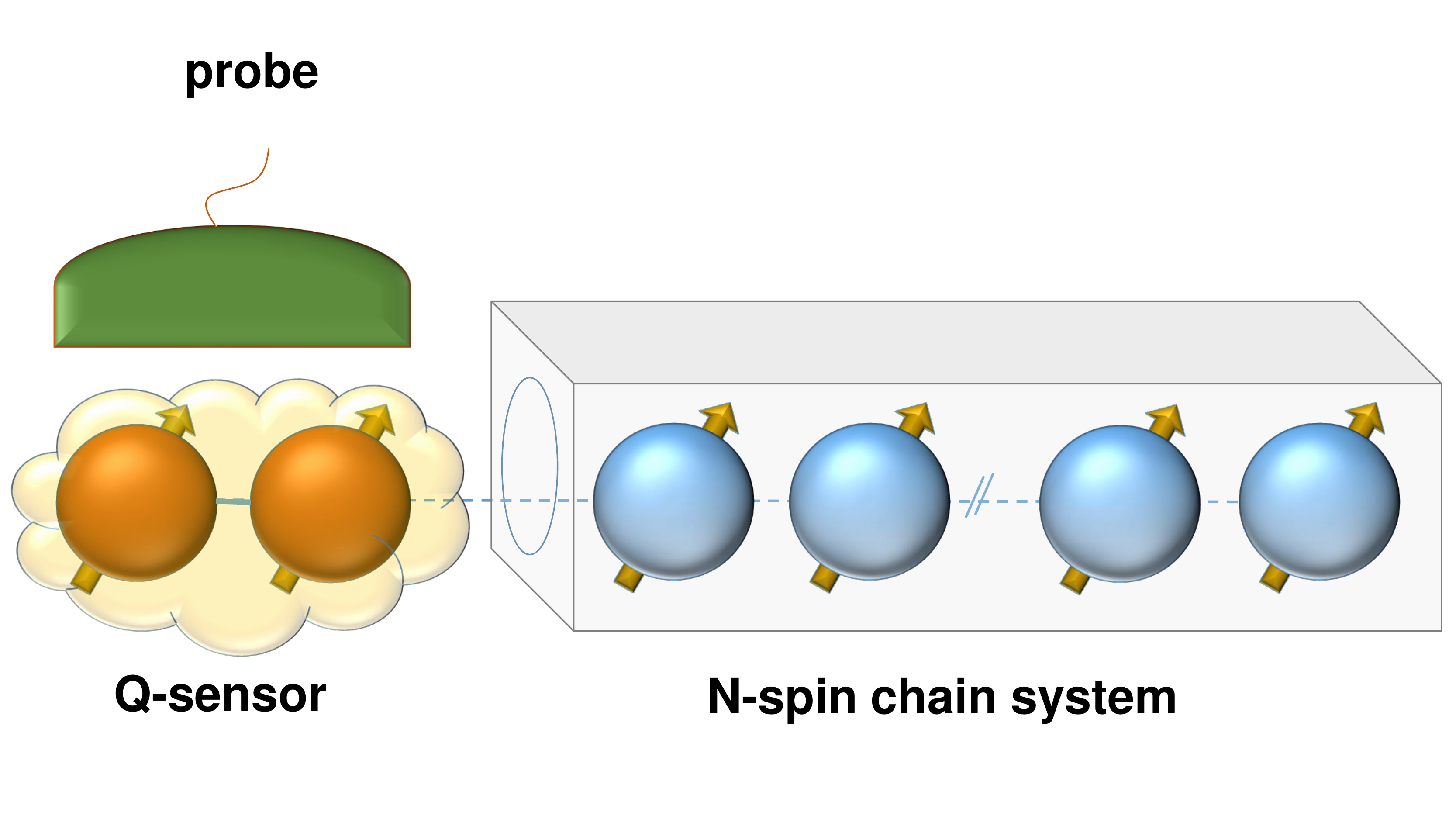}		
  	\caption{Schematic of a spin chain system connected a two-qubit sensor. The Q-sensor, consisting of two qubits, is measured by a probe. The object system consists of $N$ spins. }
  	\label{2-spins sensor}		
  \end{figure}

We consider a spin chain system consisting of $N$ qubits with the Hamiltonian
\begin{equation}\label{hamiltonian}
H=\sum_{k=1}^{N-1} h_k H_k,
\end{equation}
where the $H_k$ are known Hermitian operators and the $h_k$ are unknown coupling constants (strengths) to be estimated \cite{junzhang2014}. Here, we mainly concentrate on the magnitude of these unknown parameters and aim to estimate their magnitudes. In this paper, we further specify the system to be the exchange model without transverse field \cite{Di2008,Chris2005} such that
\begin{equation}\label{hamiltonianmodel}
H_{k}= \frac{h_k}{2}(X_k X_{k+1} + Y_{k}Y_{k+1}),
\end{equation}
where the subscript $k$ indicates the $k$-th spin. 
 To write the operators in a compact form, we omit the tensor product symbol and the identity operator unless otherwise specified.

For a single-qubit sensor in Fig. \ref{1-spin sensor}, we have
\begin{equation}\label{Hamiltoniansinglequbit}
H=\frac{h_\beta}{2}(X_\beta X_{1} + Y_{\beta}Y_{1})+\sum_{k=1}^{N-1} \frac{h_k}{2}(X_k X_{k+1} + Y_{k}Y_{k+1}),
\end{equation}
where the subscript $\beta$ indicates the sensor and $\frac{h_\beta}{2}(X_\beta X_{1} + Y_{\beta}Y_{1})$ is the interaction Hamiltonian indicating how the sensor and the chain system are coupled.

For a two-qubit sensor in Fig. \ref{2-spins sensor}, we have
\begin{equation}\label{twosensorHamiltonian}
\begin{split}
H=& \frac{h_\alpha}{2}(X_\alpha X_{\beta} + Y_{\alpha}Y_{\beta})+\frac{h_\beta}{2}(X_\beta X_{1} + Y_{\beta}Y_{1}) \\
& +\sum_{k=1}^{N-1} \frac{h_k}{2}(X_k X_{k+1} + Y_{k}Y_{k+1}),
\end{split}
\end{equation}
where $\frac{h_\alpha}{2}(X_\alpha X_{\beta} + Y_{\alpha}Y_{\beta})$ is the internal Hamiltonian of the two qubits of the Q-sensor while $\frac{h_\beta}{2}(X_\beta X_{1} + Y_{\beta}Y_{1})$ is the interaction Hamiltonian of the sensor and the chain system. The subscript $\alpha$ and $\beta$ indicates the first and second qubits of the sensor, respectively. $h_\alpha$ can be fabricated and we assume that it is known to us while $h_\beta$ could be either unknown or known. The sensor is employed to identify all of the unknown parameters in $\{h_i\}$.

\subsection{Problem formulation}\label{2.2problemformulation}
The capability of a sensor relates to the following aspects: the structures of the sensor and the object system, the measurement scheme and the initial settings. A good sensing scheme should ensure that all of the unknown parameters appear in the matrix $A$ with a proper structure and all of the parameters can be estimated using the measurement data. Here, we consider the capability of different sensing schemes for both the single-qubit sensor and the two-qubit sensor.

To determine the sensing capability, a dynamic model of the whole system including both the sensor and the object spin chain system should be obtained. A good dynamic model can benefit the process of determining capability. Here, we employ the state space model proposed in \cite{junzhang2014}.

Given the system Hamiltonian $H$, the time evolution of a system observable $O(t)$ in the Heisenberg picture is
\begin{equation}\label{Sfunction}
\frac{dO(t)}{dt}=\mi[H,O(t)],
\end{equation}
where $\mi$ is the imaginary unit (setting the Planck constant $\hbar=1$) \cite{QMGriff}.
Let $M$ be the measurement operator. The system operators that are coupled to $M$, together with $M$, form the accessible set $G=\{O_1 \   O_2 \  O_3 \  \cdots\ \}$. See \cite{qiyu2019accessibleset} for a detailed procedure for producing the accessible set. Define the system state $\bold{x}$ as the vector of all expectations of the operators in the accessible set $G$, then we have
\begin{equation}\label{state}
\bold{x}=(\hat{O}_1 \ \   \hat{O}_2 \  \cdots\  \hat{O}_i \ \cdots)^T.
\end{equation}
where $\hat{O}_i=\text{Tr}(O_i\rho)$ is the expectation value of observable $O_i\in G$. Intuitively, elements in $G$ are those influenced by the measurement operator during the time evolution. The linear equations describing the dynamics of the system with initial state $\bold{x}(0)=\bold{x}_0$ can be written as
\begin{equation}\label{statespaceeq}
		\left\{
\begin{array}{rl}
\dot{\bold{x}} &=A\bold{x}+Bu, \\
y &=C\bold{x},
\end{array}
\right.
\end{equation}
where $B=\bold{x}_0$ and $u=\delta(t)$.
Given the system state $\bold{x}$ defined in \eqref{state}, the derivative of $\bold{x}$ can be obtained by \eqref{Sfunction}, based on which the $A$ matrix can then be calculated. Setting the initial state of the system, which is equivalent to set the initial control signal, the matrix $B$ is then determined. $C$ is the decomposition coefficients of the measurement matrix onto the operators in the state $\bold{x}$. In \cite{junzhang2014}, the authors pointed out that a necessary condition for reconstructing the system is that all of the unknown parameters in $\{h_i\}$ appear in the matrix $A$.

For cases in which the system is of an arbitrary dimension, the capability of the sensor can be determined using the STA method given the state space model. Otherwise given a fixed system dimension, we can apply an approach involving transferring the state space model into its corresponding transfer function and then employing methods like the Gr$\ddot{\text{o}}$bner basis method. The relationship between a transfer function and state space functions in \eqref{statespaceeq} is
\begin{equation}\label{StateSpace2Transfer}
G(s)=\bold{C}(s\bold{I}-\bold{A})^{-1}\bold{x}_0,
\end{equation}
where $s\in \mathbb{C}$ is the Laplace variable. Note that, the relationship between the transfer function and the state space equations is not bijective, but determined only up to a similarity transformation. There is a unique transfer function for a given state space matrix while there are many different state space realizations for a known transfer function.


In practical situations, it can be difficult to initialize a qubit to an arbitrary state or to measure the system using an arbitrary operator (e.g., in a solid-state qubit system). Special constraints are applied to individual experiments. Here, we consider the case in which the initial state can only be prepared in eigenstates of the $X$ operator while the measurement is confined to be the $Y$ or $Z$ operator. In this paper, these constraints apply to both a single-qubit sensor and a two-qubit sensor.

For a single-qubit sensor, we have the following initial state
\begin{equation}\label{initialstatesiglequbit}
\rho_{ini}= \rho^x_\beta ,
\end{equation}
where $ \rho^x_\beta=\frac{I+\sigma_x}{2}$ is the eigenstate of $X_\beta$. The subscript $\beta$ labels the sensor and the superscript $x$ indicates the state is an eigenstate of the $X$ operator.

We consider two measurement schemes
\begin{equation}\label{MeasurementSetsinglequbit}
\mathbb{M}_1=\{ Y_\beta, \ Z_\beta\}.
\end{equation}
The corresponding accessible sets are listed below:
\begin{enumerate}
	\item [1)] $\tilde{M}^1=Y_\beta$, \\
	 $\tilde{G}^1=\{Y_\beta,\   Z_\beta X_1, \  Z_\beta Z_1Y_2, \  Z_\beta Z_1Z_2X_3,\ \cdots\  \}$,
	\item [2)] $\tilde{M}^2=Z_\beta$, \\
	 $\tilde{G}^2=\{Z_\beta,\  Y_\beta X_1,\  X_\beta Y_1, \ Y_\beta Z_1X_2,\  X_\beta Z_1Y_2,\  \cdots \ \}$.
\end{enumerate}
Note that we restrict the initial state to be prepared as eigenstates of the $X_\beta$ operator. However, $X_\beta$ belongs to neither the accessible set $\tilde{G}^1$ nor $\tilde{G}^2$, which means it is not coupled to any operators in the above accessible sets $\tilde{G}^1$ and $\tilde{G}^2$. In this situation, the expectations of measurements $Y$ and $Z$ on the sensing qubits are always zero, which means that for both cases we have $\bold{x_0}=0$ and $B=0$. From \eqref{statespaceeq} we always have $\bold{x}\equiv 0$ which results in the measurement $y\equiv 0$ for both measurement schemes $\tilde{M}^1$ and $\tilde{M}^2$. Hence, no information can be extracted in these situations. Therefore, the sensor loses the capability to identify the unknown parameters of the spin chain system. This observation is summarized as follows.
\begin{proposition}\label{SingleQubit}
For a quantum system with Hamiltonian given in \eqref{Hamiltoniansinglequbit}, a single-qubit sensor is incapable of identifying all of the unknown parameters in the system Hamiltonian when the initial state of the sensor is given in  \eqref{initialstatesiglequbit} and the measurement scheme is selected from \eqref{MeasurementSetsinglequbit}.
\end{proposition}

For the two-qubit sensor, the initial state can be chosen from the following set
\begin{equation}\label{initialstate}
\{ \rho^x_\alpha\otimes \frac{I}{2},\  \frac{I}{2}\otimes\rho^x_\beta,\  \rho^x_\alpha\otimes \rho^x_\beta  \}.
\end{equation}
The measurement operator is chosen from the following set
\begin{equation}\label{MeasurementSet}
\begin{split}
\mathbb{M}=&\{ Y_\alpha I_\beta,\ Z_\alpha I_\beta, Y_\alpha Z_\beta,\  Y_\alpha Y_\beta,\\
&Z_\alpha Y_\beta,\  Z_\alpha Z_\beta,\ I_\alpha Y_\beta,\ I_\alpha Z_\beta\},
\end{split}
\end{equation}
which indicates that we can only prepare the probe to measure the system using the $Y$ and $Z$ operators. For example, using the measurement operator $Y_\alpha Z_\beta$ means that we measure $Y$ on the first sensing qubit and $Z$ on the second sensing qubit. The cases $M=Y_\alpha I_\beta$ and $M=Z_\alpha I_\beta$ are similar to the cases $\tilde{M}^1$ and $\tilde{M}^2$ of the single-qubit sensor, which lead to the conclusion that these two measurement schemes are not capable of estimating the unknown parameters $\{h_i\}$.

The accessible sets for different measurement schemes are given as follows:
\begin{enumerate}
	\item [1)] $M^1=Z_\alpha Y_\beta$, \\
	 $G^1=\{X_\alpha ,\ Z_\alpha Y_\beta,\ Z_\alpha Z_\beta X_1,\ Z_\alpha Z_\beta Z_1Y_2,\ \cdots \ \}$;
	\item [2)] $M^2=Y_\alpha Z_\beta$, \\
	 $G^2=\{Y_\alpha Z_\beta ,\ X_\beta,\ Z_\beta Y_1,\ Y_\alpha X_\beta Y_1,\ X_\alpha Y_\beta Y_1,\ Y_\alpha Z_1,\ \cdots \ \}$;
	\item [3)] $M^3=Y_\alpha Y_\beta$, \\
	 $G^3=\{Y_\alpha Y_\beta, \ Y_\alpha Z_\beta X_1,\  X_\beta X_1, \ Y_\alpha Z_\beta Z_1 Y_2,\  X_\beta Z_1 Y_2, \ \cdots \ \}$;
	\item [4)] $M^4=Z_\alpha Z_\beta$, \\
	 $G^4=\{Z_\alpha Z_\beta ,\ X_\alpha Y_\beta Z_1,\ X_\alpha X_1,\ Z_\alpha Y_\beta X_1,\ Z_\alpha Z_1,\ \cdots \ \}$;
	\item [5)] $M^5=I_\alpha Y_\beta$, \\
	$G^5=\{ Y_\beta,\ X_\alpha Z_\beta,\ X_\alpha Y_\beta X_1,\ X_\alpha X_\beta Y_1,\ Z_\alpha X_1,\ \cdots\}$;
	\item [6)] $M^6=I_\alpha Z_\beta$, \\
	$G^6=\{ Z_\alpha,\ X_\alpha Y_\beta,\ Z_\beta,\ Y_\alpha X_\beta,\ Y_\beta X_1,\ \cdots\}$,
\end{enumerate}
where $M^i,\  i\in\{1,2,3,4, 5, 6\}$ represent six different measurement schemes and $G^i, i\in\{1,2,3,4, 5, 6\}$ are the corresponding accessible sets. When measuring $Y_\alpha Y_\beta$ and $Z_\alpha Z_\beta$, all of the initial states in $\rho_{ini}$ are orthogonal to operators in the accessible sets $G^3$, $G^4$, $G^5$ and $G^6$. Here, the quantum sensor fails to acquire any information of the unknown parameters. Therefore, we will only consider the remaining two cases: the probe is $Z_\alpha Y_\beta$ with the initial state $\rho^x_\alpha\otimes \frac{I}{2}$; the probe is $Y_\alpha Z_\beta$ with the initial state $\frac{I}{2}\otimes\rho^x_\beta$. For the generation of the accessible sets $G^1$ and $G^2$, please see Appendix \ref{APPX1}.

To study the capability of the two-qubit sensor, the following two problems will be addressed.
\begin{problem}\label{problem1}
	For a quantum spin chain system of the exchange model without transverse field with the Hamiltonian given in \eqref{twosensorHamiltonian}, when the measurement is $Z_\alpha Y_\beta$ and the initial state is $\rho^x_\alpha\otimes \frac{I}{2}$, is the information extracted from the sensor enough to identify the coupling strengths in the spin system?
\end{problem}
\begin{problem}\label{problem2}
For a quantum spin chain system of the exchange model without transverse field with the Hamiltonian given in \eqref{twosensorHamiltonian}, when the measurement is $Y_\alpha Z_\beta$ and the initial state is $\frac{I}{2}\otimes\rho^x_\beta$, is the information extracted from the sensor enough to identify the coupling strengths in the spin system?
\end{problem}
The above two problems are answered in Section \ref{section3} and Section \ref{section4}, respectively.

\section{Capability test when  measuring $Z_\alpha Y_\beta$}\label{section3}
\subsection{STA method}\label{STA}
The similarity transformation approach (STA) was employed in \cite{yuanlong2020} to determine the quantum system identifiability. We employ this method to determine the sensor capability. In our case, if there exist more than one realizations of the parameters $\{h_i\}$ that can generate the same input-output behavior, then we claim that the sensor is unable to identify the coupling parameters.

The criterion equations for the STA method were given in \cite{yuanlong2020} as follows
{\setstretch{0.3}
\begin{equation}\label{STAequations01}
SA(h) =A(h^\prime)S,
\end{equation}
\begin{equation}\label{STAequations02}
S\bold{x}_0=\bold{x}_0,
\end{equation}
\begin{equation}\label{STAequations03}
C=CS,
\end{equation}}
where $A,\ B,\ C$ are matrices of state space functions in \eqref{statespaceeq} and $S$ is a non-singular similarity transformation matrix. For a minimal system, if all solutions to the above equations satisfy $|h_i|=|h_i^\prime|$, then the sensor is able to identify these parameters. Otherwise, it is incapable since the uniqueness of the parameter set $\{h_i\}$ can not be guaranteed. 
For a non-minimal system, the structure preserving transformation (SPT) method was proposed in \cite{yuanlong2020} to obtain a minimal subsystem which can inherit some good structural features from the original state space model. The STA method can be applied with the assistance of SPT method for a class of non-minimal systems \cite{yuanlong2020}.

\subsection{Proof of capability}
For \textbf{Problem 1}, the accessible set is
\begin{equation}\label{eq8}
G^1=\{X_\alpha,\ Z_\alpha Y_\beta,\ Z_\alpha Z_\beta X_1, \ Z_\alpha,\  Z_\beta Z_1 Y_2, \ Z_\alpha Z_\beta Z_1 Z_2 X_3,\ \cdots \ \}.
\end{equation}
Different orders of the operators can yield different $A$, $B$ and $C$ matrices which represent the same system. A good ordering of element operators can result in state space equations with good structure, which can benefit the determination process for sensing capability. Here, we choose it to be in the order as shown in Appendix \ref{ASZ1Y2}, which is the same as \eqref{eq8}. The initial state is chosen as $\rho^x_\alpha\otimes \frac{I}{2}$ and we thus have
\begin{equation}\label{BmatrixZ1Y2}
B=(1 \ 0\  0\  0 \  0 \   \cdots \ 0 \ )^T.
\end{equation}
The corresponding matrix $A$  is
\begin{equation}\label{AmatrixZ1Y2}
A =  \left(
\begin{array}{cccccc}
0 & h_\alpha & 0 & 0 & 0 & \cdots\\
-h_\alpha & 0 & h_\beta & 0 & 0 & \cdots\\
0 & -h_\beta & 0 & h_1 & 0 & \cdots\\
0 & 0 & -h_1 & 0 & \ddots & \cdots \\
0 & 0& 0&  \ddots & \quad &h_{N-1} \\
 \vdots & \vdots & \vdots & \quad &-h_{N-1} & 0\\
\end{array}
\right). \\
\end{equation}
$A$ is of dimension $(N+2)\times(N+2)$.
We have the following conclusion whose proof is presented in Appendix \ref{APPlem1}.
\begin{lemma}\label{lem1}
With (\ref{AmatrixZ1Y2}) and $B=(1,\ 0,\ \cdots,\ 0)^T$, the controllability matrix $\text{CM}=(B,\ AB,\ \cdots,\ A^{(N+1)}B)$ has full rank for almost any value of $\{h_i\}$.
\end{lemma}
Considering measuring $Z_\alpha Y_\beta$, we have
\begin{equation}\label{CmatrixZ1Y2}
C=(0\ 1\ 0\ \cdots\ 0).
\end{equation}
We present the following theorem which provides an answer to Problem \ref{problem1}.
\begin{theorem}\label{th4}
	For a system with state space equations given in \eqref{BmatrixZ1Y2}, \eqref{AmatrixZ1Y2} and \eqref{CmatrixZ1Y2}, where $h_\alpha$ and $h_\beta$ can be designed, the unknown parameters $\{h_i\}$ can be estimated from the measurement results of $Z_\alpha Y_\beta$.
\end{theorem}
\begin{proof}
According to the parity of the dimension, we prove the theorem separately. First we consider that $N$ is even. \\


	1) $N$ is even \\

	When $N$ is even, we first prove that the system described by \eqref{BmatrixZ1Y2}, \eqref{AmatrixZ1Y2} and \eqref{CmatrixZ1Y2} is minimal for almost any value of the unknown parameters. Then we prove that the sensor is capable of identifying the coupling parameters when $N$ is even.
	
We first prove the \textbf{minimality}.
	
	\begin{lemma}\label{lem2}
		Let $N$ be even. With \eqref{AmatrixZ1Y2} and \eqref{CmatrixZ1Y2}, the observability matrix
		\begin{equation}
		 \text{OM}=\left(
		\begin{array}{c}
		C\\
		CA\\
		\vdots\\
		CA^{N+1}\\
		\end{array}\nonumber\right)
		\end{equation}
		has full rank for almost any value of $\{h_i\}$.
	\end{lemma}
	The proof of Lemma \ref{lem2} is presented in Appendix \ref{APPlem2}.
	According to Lemmas \ref{lem1} and \ref{lem2}, the system is minimal for almost any value of the unknown parameters. Now we perform the \textbf{capability test.}
	Using (\ref{STAequations02}) and (\ref{STAequations03}) we know that $S$ is of the form
	\begin{equation}
	S=\left(
	\begin{array}{ccccc}
	1 & * & * &\cdots & *\\
	0 & 1 & 0 &\cdots & 0\\
	0 & * & * &\cdots & *\\
	\vdots & \vdots & \vdots & & \vdots\\
	0 & * & * &\cdots & *\\
	\end{array}
	\right)_{(N+2)\times(N+2)},
	\end{equation}	
	and (\ref{STAequations01}) is now
	\begin{equation}\label{eq20}
	\begin{array}{rl}
	&\left(
	\begin{array}{cccc}
	1 & * & \cdots & *\\
	0 & 1 & \cdots & 0\\
	\vdots & \vdots &  & \vdots\\
	0 & * & \cdots & *\\
	\end{array}
	\right)
	\left(
	\begin{array}{ccccc}
	0 & h_\alpha & 0 & \cdots & 0\\
	-h_\alpha & 0 & \ddots & &\\
	0 & \ddots &  & &\\
	\vdots &  &  & &\\
	\end{array}
	\right)\\
	=&\left(
	\begin{array}{ccccc}
	0 & h'_\alpha & 0 & \cdots & 0\\
	-h'_\alpha & 0 & \ddots & &\\
	0 & \ddots &  & &\\
	\vdots &  &  & &\\
	\end{array}
	\right)
	\left(
	\begin{array}{cccc}
	1 & * & \cdots & *\\
	0 & 1 & \cdots & 0\\
	\vdots & \vdots &  & \vdots\\
	0 & * & \cdots & *\\
	\end{array}
	\right).\\
	\end{array}
	\end{equation}
	In this paper, the $*$ represents an indeterminate element. This case can be proved by following a similar approach to that the proof of Theorem $1$ in \cite{yuanlong2020}.
	
	Denote the product matrix on the LHS and RHS of (\ref{eq20}) as $L$ and $R$, respectively. From $L_{21}=R_{21}$, we have $\theta_1=\theta_1'$. From $L_{1\sigma}=R_{1\sigma}$, we have $S_{1\sigma}=(1,0,...,0)$. For an arbitrary matrix $A$, we denote $A_{\sigma i}$ and $A_{j \sigma}$ as its $i$-th column and $j$-th row, respectively. Then $S$ is of the form
	\begin{equation}
	S=\left(
	\begin{array}{ccccc}
	1 & 0 & 0 &\cdots & 0\\
	0 & 1 & 0 &\cdots & 0\\
	0 & * & * &\cdots & *\\
	\vdots & \vdots & \vdots & & \vdots\\
	0 & * & * &\cdots & *\\
	\end{array}
	\right)_{(N+2)\times(N+2)},
	\end{equation}	
	
	If we denote the partitioned $S$ and $A$ as $$S=\left(\begin{matrix}1_{1\times 1}&\mathbf{0}_{1\times (N+1)}^T\\ \mathbf{0}_{(N+1)\times 1}&\tilde{S}_{(N+1)\times (N+1)}\end{matrix}\right),$$$$ A=\left(\begin{matrix}0_{1\times 1}& \mathbf{E}_{1\times (N+1)}^T \\ -\mathbf{E}_{(N+1)\times 1} &\tilde{A}_{(N+1)\times (N+1)}\end{matrix}\right),$$
	then (\ref{eq20}) is equivalent to
	\begin{equation}\label{eq21}
	\mathbf{E}^T=\mathbf{E}'^T\tilde S,
	\end{equation}
	\begin{equation}\label{eq22}
	-\tilde S \mathbf{E}=-\mathbf{E}',
	\end{equation}
	\begin{equation}\label{eq23}
	\tilde S\tilde A=\tilde A'\tilde S.
	\end{equation}
	From the first elements in (\ref{eq21}) and (\ref{eq22}), we have $h_\alpha=h'_\alpha\tilde{S}_{11}$ and $-\tilde{S}_{11}h_\alpha=-h'_\alpha$. Since the atypical case \cite{yuanlong2018TAC}  of $h_\alpha=0$ can be avoided by experimental design and the case is then omitted, we have $h'_\alpha\neq0$ and $|\tilde{S}_{11}|=1$, which indicates $|h_\alpha|=|h'_\alpha|$. From the remaining elements in (\ref{eq21}) and (\ref{eq22}), we have $\tilde{S}_{12}=\tilde{S}_{13}=\cdots=\tilde{S}_{1n}=0$ and $\tilde{S}_{21}=\tilde{S}_{31}=\cdots=\tilde{S}_{n1}=0$.
	
	If $\tilde{S}_{11}=1$, (\ref{eq23}) now is of the same form as (\ref{eq20}) but with the dimension decreased by $1$; otherwise if $\tilde{S}_{11}=-1$, (\ref{eq23}) is equivalent to $(-\tilde S)\tilde A=\tilde A'(-\tilde S)$, which is also of the same form as (\ref{eq20}) with the dimension decreased by $1$. Therefore these procedures can be performed inductively and finally we know all of the solutions to (\ref{eq20}) satisfy $S=\text{diag}(1,\ \pm1,\ \cdots,\ \pm1)$ and $|h_\alpha|=|h'_\alpha|$, $|h_\beta|=|h'_\beta|$ and $|h_i|=|h_i'|$ for all $1\leq i\leq (N-1)$. When $N$ is even, the sensor is capable of identifying the coupling parameters.\\
	
		2) $N$ is odd \label{sst1} \\
		
		Now we consider the case of odd $N$. We first prove that the system is always non-minimal. Then the SPT method can be employed to obtain its minimal realization while trying to maintain most of its structure. Then we prove that all of the unknown parameters can be estimated.
		
		When $N$ is odd, although the system is still controllable using Lemma \ref{lem1}, we have the following lemma (the proof is presented in Appendix \ref{APPlem3}).
		\begin{lemma}\label{lem3}
			Let $N$ be odd. With (\ref{AmatrixZ1Y2}) and $C=(0,\ 1,\ 0,\ \cdots,\ 0)$, the system is always unobservable.
		\end{lemma}
		
		We employ the SPT method to prove the conclusion \cite{yuanlong2020}. The first step is to perform \textbf{minimal decomposition} and in this case considering the observability is enough. We construct
		\begin{equation}\label{eqad1}
		P=\left(
		\begin{array}{cccc}
		&C& &\\
		&CA& &\\
		&\vdots& &\\
		&CA^{N}& &\\
		0&\cdots&0&1\\
		\end{array}\right)
		\end{equation}
		and partition it as
		\begin{equation}\label{eqad2}
		P=\left(\begin{matrix} \bar{P}_{(N+1)\times (N+1)} & \mathbf{p}_{(N+1)\times1}\\ \mathbf{0}^T_{1\times (N+1)} & 1  \end{matrix}\right).
		\end{equation}
		
		$\bar{P}$ is the observability matrix whose dimension equals to the system dimension decreased by one. Using Lemma \ref{lem2}, we know $\bar{P}$ is non-singular, and therefore $P$ is non-singular. Since the first $(N+1)$ rows of $P$ are linearly independent, we know $P$ is the similarity transformation matrix which can decompose the system into observable and unobservable parts. Furthermore, its unobservable subspace is only one-dimensional.
		
		Then we choose the second transformation as \\
		$\bar{P}^{-1}_{(N+1)\times (N+1)}\oplus I_{1\times1}$. The total transformation matrix is $Q=\left(\begin{matrix} \bar{P}^{-1} & \mathbf{0}\\ \mathbf{0}^T & 1  \end{matrix}\right)\left(\begin{matrix} \bar{P} & \mathbf{p}\\ \mathbf{0}^T & 1  \end{matrix}\right)=\left(\begin{matrix} I& \bar{P}^{-1}\mathbf{p} \\ \mathbf{0}^T & 1  \end{matrix}\right)$, and its inversion is $Q^{-1}=\left(\begin{matrix} I& -\bar{P}^{-1}\mathbf{p} \\ \mathbf{0}^T & 1  \end{matrix}\right)$. Let $\mathbf{\bar{x}}=Q\mathbf{{x}}$ generate the system $\bar{\Sigma}=(\bar{A},\bar{B},\bar{C})$:
		\begin{equation}\label{eqad38}
		\left\{
		\begin{array}{rl}
		\mathbf{\dot{\bar{x}}}&=\bar{A}\mathbf{\bar{x}}+\bar{B}{\delta(t)},\ \ \mathbf{\bar{x}}(0)=\mathbf{0},\\
		y&=\bar{C}\mathbf{\bar{x}}.\\
		\end{array}
		\right.
		\end{equation}
		We partition $A$ as
		\begin{equation}\label{eq39}
		A=\left(\begin{matrix}
		UL_{(N+1)\times (N+1)} &\ UR_{(N+1)\times 1}\\
		DL_{1\times (N+1)} &\ DR_{1\times 1}\\
		\end{matrix}\right).
		\end{equation}
		Then we have
		\begin{equation}
		\begin{array}{rl}
		\bar{A}&=QAQ^{-1}=\left(\begin{matrix} I& \bar{P}^{-1}\mathbf{p} \\ \mathbf{0}^T & 1  \end{matrix}\right)\left(\begin{matrix}
		UL &\ UR\\
		DL &\ DR\\
		\end{matrix}\right)\left(\begin{matrix} I& -\bar{P}^{-1}\mathbf{p} \\ \mathbf{0}^T & 1  \end{matrix}\right)\\
		&=\left(\begin{matrix}
		UL+\bar{P}^{-1}\mathbf{p}DL &\ *_{(N+1)\times1}\\
		*_{1\times (N+1)} &\ *_{1\times1}\\
		\end{matrix}\right),\\
		\end{array}
		\end{equation}
		and
		\begin{equation}
		\begin{split}
		\bar{B}&=QB=Q_{\sigma1}=(1,\ 0,\ \cdots,\ 0)^T, \\
		\bar{C}&=CQ^{-1}=(Q^{-1})_{2\sigma}=(0,\ 1,\ 0,\ \cdots,\ 0,*).
		\end{split}
		\end{equation}
		Partition $\mathbf{\bar{x}}=(\mathbf{\tilde{x}}^T,\ \bar{x}_{N+2})^T$. From the construction of $Q$ we know $\bar{\Sigma}$ is in the observable canonical form. Therefore, $\mathbf{\tilde{x}}$ corresponds to the minimal subspace (denoted as $\tilde{\Sigma}$) of $\bar{\Sigma}$. We denote the evolution of $\mathbf{\tilde{x}}$ as:
		\begin{equation}\label{eq291}
		\left\{
		\begin{array}{rl}
		\mathbf{\dot{\tilde{x}}}&=\tilde{A}\mathbf{\tilde{x}}+\tilde{B}{\delta(t)},\ \ \mathbf{\tilde{x}}(0)=\mathbf{0},\\
		y&=\tilde{C}\mathbf{\tilde{x}},\\
		\end{array}
		\right.
		\end{equation}
		and then we have
		\begin{equation}
		\begin{split}
		\tilde{A}&=UL+\bar{P}^{-1}\mathbf{p}DL, \\
		\tilde{B}&=(1,\ 0,\ \cdots\ ,\ 0)^T, \\
		\tilde{C}&=(0,\ 1,\ 0,\ \cdots,\ 0).
		\end{split}
		\end{equation}
		Hence, $\tilde{B}$ and $\tilde{C}$ are obtained by removing the last element of $B$ and $C$, respectively.
		
		To further determine $\tilde{A}$, we use the following lemma whose proof is presented in Appendix \ref{APPlem4}.
		\begin{lemma}\label{lem4}
			Let $N$ be odd. Given \eqref{AmatrixZ1Y2} and \eqref{CmatrixZ1Y2}. With (\ref{eqad1}) and (\ref{eqad2}), it holds that
			$\mathbf{p}=(0,\ \cdots,\ 0,\ h_\beta \prod_{i=1}^{N-1} h_i)^T$ and
			\begin{equation}
			\bar{P}^{-1}=\left(
			\begin{array}{ccc}
			\cdots & * & -K \\
			\cdots & * & 0 \\
			\cdots & * & -Kh_\alpha/h_\beta \\
			\cdots & * & 0 \\
			\vdots & \vdots & \vdots \\
			\cdots & * & \frac{-Kh_\alpha\prod_{i=1}^{(N-1)/2-1}h_{2i-1}}{(h_\beta\prod_{i=1}^{(N-1)/2-1}h_{2i}) } \\
			\cdots & * & 0 \\
			\end{array}
			\right)/\det(\bar{P}),
			\end{equation}
			where $\det(\bar{P})=
			h_\alpha h_{\beta}^{N-1}h_{N-2}^3\prod_{i=1}^{(N-3)/2}h_{2i-1}^{N+2-2i}h_{2i}^{(N-1)-2i}$ and $$K=\left\{\begin{matrix}
			1,& N=1\\
			h_{\beta}^{N-1}\prod_{i=1}^{N-2}h_i^{(N-1)-i},& N\geq 3.\\
			\end{matrix}\right.$$
		\end{lemma}
		
		Using Lemma \ref{lem4}, we have
		\begin{equation}\label{eq42}
		\begin{array}{rl}
		&\tilde{A}=UL+\bar{P}^{-1}\mathbf{p}DL\\
		&=UL+\bar{P}^{-1}_{(N+1)\times (N+1)}\cdot \\
		&\ \ \ [(0,\cdots,0,h_\beta \prod_{i=1}^{N-1} h_i)^T]_{(N+1)\times1}(0,\cdots,0,-h_{N-1})_{1\times (N+1)}\\
		&=UL+(h_\beta \prod_{i=1}^{N-1} h_i)(\bar{P}^{-1})_{\sigma (N+1)}(0,\cdots,0,-h_{N-1})_{1\times (N+1)}\\
		&=UL-h_{N-1}(h_\beta \prod_{i=1}^{N-1} h_i)[0_{(N+1)\times N},\ \ (\bar{P}^{-1})_{\sigma (N+1)}],\\
		\end{array}
		\end{equation}
		and further
		\begin{equation}\label{eq43}
		\tilde{A}_{k(N+1)}=\left\{
		\begin{matrix}
		0,& k\ is \  even\\
		\frac{h_\beta h_{N-1}\prod_{i=1}^{(N-1)/2}h_{2i}}{h_\alpha \prod_{i=1}^{(N-1)/2}h_{2i-1}},   & k=1\\
		\frac{h_{N-1}\prod_{i=(k-1)/2}^{(N-1)/2}h_{2i}}{\prod_{i=(k-1)/2}^{(N-1)/2}h_{2i-1}} ,& k \  is \  odd \  and \ 1< k< N+1\\
		h_{N-2}+\frac{h_{N-1}^2}{h_{N-2}} ,& k=N+1.\\
		\end{matrix}
		\right.
		\end{equation}
		
		We compare \eqref{eq42} with \eqref{eq39} and find the following property. If we remove the downmost $(N+2)$th row and the rightmost $(N+2)$th column of $A$ and change its $(N+1)$th column to be \eqref{eq43}, we obtain $\tilde{A}$, and the first $N$ columns of $\tilde{A}$ are the same as these of $A$. Hence, we maintain most of the structure properties in $A$, which makes it possible to analyse the capability of the sensor using STA equations.
		
		Having proven that the system for odd $N$ is not minimal and obtained a minimal subsystem \eqref{eq291}, we analyse the sensing capability using the STA method \cite{yuanlong2020}. Using (\ref{STAequations02}) and (\ref{STAequations03}) we know $S$ is of the form
		\begin{equation}
		S=\left(
		\begin{array}{ccccc}
		1 & * & * &\cdots & *\\
		0 & 1 & 0 &\cdots & 0\\
		0 & * & * &\cdots & *\\
		\vdots & \vdots & \vdots & & \vdots\\
		0 & * & * &\cdots & *\\
		\end{array}
		\right)_{(N+1)\times (N+1)},
		\end{equation}
		and (\ref{STAequations01}) now becomes
		{
			\begin{equation}\label{eq44}
			\begin{array}{rl}
			&\left(
			\begin{array}{ccccc}
			1 & * & * &\cdots & *\\
			0 & 1 & 0 &\cdots & 0\\
			0 & * & * &\cdots & *\\
			\vdots & \vdots & \vdots & & \vdots\\
			0 & * & * &\cdots & *\\
			\end{array}
			\right)
			\addtolength{\arraycolsep}{-0.11cm}
			\left(
			\begin{array}{cccccc}
			0 & h_\alpha & 0 & \cdots & 0 &\tilde{A}_{1(N+1)}\\
			-h_\alpha & 0 & h_\beta & & 0&0\\
			0 & -h_\beta &  & \ddots & & \\
			\vdots &  & \ddots & &h_{N-3}& 0\\
			0 & \cdots & 0 & -h_{N-3}&0 & \tilde{A}_{N(N+1)}\\
			0 & \cdots & 0 &0 &-h_{N-2} & 0\\
			\end{array}
			\right)\\
			=&\left(
			\addtolength{\arraycolsep}{-0.11cm}
			\begin{array}{cccccc}
			0 & h'_\alpha & 0 & \cdots & 0 &\tilde{A}_{1(N+1)}'\\
			-h'_\alpha & 0 & h'_\beta & & 0&0\\
			0 & -h'_\beta &  & \ddots & & \\
			\vdots &  & \ddots & &h_{N-3}'& 0\\
			0 & \cdots & 0 & -h_{N-3}'&0 & \tilde{A}_{N(N+1)}'\\
			0 & \cdots & 0 &0 &-h_{N-2}' & 0\\
			\end{array}
			\right)
			\left(
			\begin{array}{ccccc}
			1 & * & * &\cdots & *\\
			0 & 1 & 0 &\cdots & 0\\
			0 & * & * &\cdots & *\\
			\vdots & \vdots & \vdots & & \vdots\\
			0 & * & * &\cdots & *\\
			\end{array}
			\right).\\
			\end{array}
			\end{equation}
		}
		
		We have the following lemma characterizing its solution:
		\begin{lemma}\label{lem5}
			For a fixed $k$ satisfying $3\leq k\leq N-1$, the solution to (\ref{eq44}) satisfies that the first $k$ columns of $S$ is
			\begin{equation}
			\left(
			\begin{array}{ccccc}
			1 & 0 & 0 & 0& \cdots \\
			0 & 1 & 0 & 0& \cdots \\
			0 & 0 & \pm1 & 0 &\cdots\\
			\vdots & \vdots & & \ddots &\\
			0 & 0 & \cdots & 0&\pm1 \\
			0 & 0 &  \cdots &0 & 0\\
			\vdots & \vdots & &\vdots  &\vdots \\
			0 & 0 &  \cdots &0 & 0\\
			\end{array}
			\right)_{(N+1)\times k},
			\end{equation}
			and we have $h_\alpha=h'_\alpha$ and $ |h_\beta|=|h'_\beta|$ and $|h_{i-2}|=|h_{i-2}'|$ for $3\leq i\leq k$.
		\end{lemma}
		
		\begin{proof}	
			Denote the product matrix on the LHS and RHS of (\ref{eq44}) as $L$ and $R$, respectively. From $L_{\sigma1}=R_{\sigma1}$, we have $-h_\alpha S_{\sigma2}=\tilde{A}_{\sigma1}'$, which implies $h_\alpha =h'_\alpha$ and $S_{\sigma2}=(0,\ 1,\ 0,\ \cdots,\ 0)^T$. From $L_{\sigma2}=R_{\sigma2}$, we have $h_\alpha S_{\sigma1}-h_\beta S_{\sigma3}=\tilde{A}_{\sigma2}'$, which implies $S_{33}h_\beta=h'_\beta$ and $S_{\sigma3}=(0,\ 0,\ S_{33},\ 0,\ \cdots,\ 0)^T$. Then from $L_{23}=R_{23}$, we have $h_\beta=h'_\beta S_{33}$. Hence, $S_{33}=\pm1$ and $|h_\beta|=|h_\beta'|$.
			
			The remaining part of Lemma \ref{lem5} can be proved via induction on $k$. Suppose it holds for $k=t\leq N-2$. Then from $L_{\sigma t}=R_{\sigma t}$, $h_{t-3}S_{\sigma (t-1)}-h_{t-2}S_{\sigma(t+1)}=S_{tt}\tilde{A}_{\sigma t}'$, which is
			\begin{equation}\label{eq47}
			\left(
			\begin{array}{c}
			\vdots\\
			-h_{t-2}S_{(t-2)(t+1)}\\
			h_{t-3}S_{(t-1)(t-1)}-h_{t-2}S_{(t-1)(t+1)}\\
			-h_{t-2}S_{t(t+1)}\\
			-h_{t-2}S_{(t+1)(t+1)}\\
			-h_{t-2}S_{(t+2)(t+1)}\\
			\vdots\\
			\end{array}
			\right)=S_{tt}\left(
			\begin{array}{c}
			\vdots\\
			0\\
			h_{t-3}'\\
			0\\
			-h_{t-2}'\\
			0\\
			\vdots\\
			\end{array}
			\right).
			\end{equation}
			From the $(t-1)$th and $(t+1)$th rows of (\ref{eq47}), we have
			\begin{equation}\label{eq48}
			h_{t-3}S_{(t-1)(t-1)}-h_{t-2}S_{(t-1)(t+1)}=S_{tt}h_{t-3}'
			\end{equation}
			and
			\begin{equation}\label{eq49}
			h_{t-2}S_{(t+1)(t+1)}=S_{tt}h_{t-2}'.
			\end{equation}
			From all of the other rows of (\ref{eq47}), we know $S_{i(t+1)}=0$ for $1\leq i\leq (N+1)$ except when $i=t-1$ or $i=t+1$.
			
			Consider $L_{(t-2)(t+1)}=R_{(t-2)(t+1)}$. The $(t-2)$th row of $S$ is $(\cdots,\ 0,\ S_{(t-2)(t-2)},\ 0,\ 0,\ 0,\ *,\ \cdots)$ and the $(t+1)$th column of $A$ is $$(\cdots,\ 0,\ h_{t-2},\ 0,\ -h_{t-1},\ 0,\ \cdots)^T.$$ Hence, $L_{(t-2)(t+1)}=0$. The $(t-2)$th row of $\tilde{A}$ is $$(\cdots,\ 0,\ -h_{t-5}',\ 0,\ h_{t-4}',\ 0,\ \cdots)$$ and the $(t+1)$th column of $S$ is $$(\cdots,\ 0,\ S_{(t-1)(t+1)},\ 0,\ S_{(t+1)(t+1)},\ 0,\ \cdots)^T.$$ Hence, $0=R_{(t-2)(t+1)}=h_{t-4}'S_{(t-1)(t+1)}$. Since the atypical case of $h_{t-4}=0$ is omitted, from $|h_{t-4}|=|h_{t-4}'|$ we know $h_{t-4}'\neq0$, which implies $S_{(t-1)(t+1)}=0$. Then (\ref{eq48}) becomes
			\begin{equation}\label{eq50}
			h_{t-3}S_{(t-1)(t-1)}=S_{tt}h_{t-3}'.
			\end{equation}
			
			Now we have proved that $S_{\sigma(t+1)}$ are all zeros except $S_{(t+1)(t+1)}$. Then the above deduction can be carried on with $t$ replaced by $t+1$, and (\ref{eq50}) now becomes
			\begin{equation}\label{eq51}
			h_{t-2}S_{tt}=S_{(t+1)(t+1)}h_{t-2}'.
			\end{equation}
			Since the atypical cases of $h_{t-3}=0$ and $h_{t-2}=0$ are excluded, (\ref{eq50}) implies $S_{tt}\neq0$, and (\ref{eq51}) implies $S_{(t+1)(t+1)}\neq0$. Hence, from (\ref{eq49}) and (\ref{eq51}) we know $S_{(t+1)(t+1)}=\pm1$ and $|h_{t-2}|=|h_{t-2}'|$, which completes the proof of Lemma \ref{lem5}.			
		\end{proof}
	
    	Lemma \ref{lem5} confirms the form of the first $N-1$ columns of the $S$ matrix. Now, we prove that the $N$th and $(N+1)$th columns take a similar form. After using Lemma \ref{lem5}, the remaining undetermined part of (\ref{eq44}) becomes
		{\setlength\arraycolsep{0.8pt}\def\arraystretch{1.5}
			\begin{equation}
			\begin{array}{rl}
			&\left(
			\begin{array}{ccccc}
			\ddots &  & \vdots &\vdots  &\vdots \\
			& \pm1 & 0 & * & * \\
			\cdots & 0 & \pm1 & * & * \\
			\cdots & 0 & 0 & * & * \\
			\cdots & 0 &  0 &* & *\\
			\end{array}\right)\!\!\!\left(\addtolength{\arraycolsep}{-0.05cm}
			\begin{array}{ccccc}
			\ddots & \ddots &  &\vdots  &\vdots \\
			\ddots & 0 & h_{N-4} & 0 & \tilde{A}_{(N-2)(N+1)} \\
			& -h_{N-4} & 0 & h_{N-3} & 0 \\
			\cdots & 0 & -h_{N-3} & 0 & \tilde{A}_{N(N+1)} \\
			\cdots & 0 &  0 &-h_{N-2} & 0\\
			\end{array}\right)\\
			\!\!&\!\!=\left( \addtolength{\arraycolsep}{-0.05cm}
			\begin{array}{ccccc}
			\ddots & \ddots &  &\vdots  &\vdots \\
			\ddots & 0 & h_{N-4}' & 0 & \tilde{A}_{(N-2)(N+1)}' \\
			& -h_{N-4}' & 0 & h_{N-3}' & 0 \\
			\cdots & 0 & -h_{N-3}' & 0 & \tilde{A}_{N(N+1)}' \\
			\cdots & 0 &  0 &-h_{N-2}' & 0\\
			\end{array}\right)\!\!\!\left(
			\begin{array}{ccccc}
			\ddots &  & \vdots &\vdots  &\vdots \\
			& \pm1 & 0 & * & * \\
			\cdots & 0 & \pm1 & * & * \\
			\cdots & 0 & 0 & * & * \\
			\cdots & 0 &  0 &* & *\\
			\end{array}\right)\!\!.\\
			\end{array}
			\end{equation}
		}
		From $L_{\sigma (N-1)}=R_{\sigma (N-1)}$, $h_{N-4}S_{\sigma (N-2)}-h_{N-3}S_{\sigma N}=S_{(N-1)(N-1)}\tilde{A}_{\sigma (N-1)}'$, which is
		{\addtolength{\arraycolsep}{-0.22cm}
			\begin{equation}\label{eq54}
			\!\!\left(
			\begin{array}{c}
			\vdots\\
			-h_{N-3}S_{(N-3)N}\\
			h_{N-4}S_{(N-4)(N-4)}-h_{N-3}S_{(N-2)N}\\
			-h_{N-3}S_{(N-1)N}\\
			-h_{N-3}S_{NN}\\
			-h_{N-3}S_{(N+1)N}\\
			\end{array}
			\right)\!\!\!=\!S_{(N-1)(N-1)}\!\!\left(
			\begin{array}{c}
			\vdots\\
			0\\
			h_{N-4}'\\
			0\\
			-h_{N-3}'\\
			0\\
			\end{array}
			\right)\!\!.
			\end{equation}
		}
		Hence, we have
		\begin{equation}\label{eq55}
		h_{N-4}S_{(N-4)(N-4)}-h_{N-3}S_{(N-2)N}=S_{(N-1)(N-1)}h_{N-4}',
		\end{equation}
		\begin{equation}\label{eq56}
		h_{N-3}S_{NN}=S_{(N-1)(N-1)}h_{N-3}',
		\end{equation}
		and $S_{iN}=0$ for $i=1,\ \cdots,\ N-3,\ N-1,\ N+1$. Now consider
		\begin{equation*}
		\begin{split}
		L_{(N-1)(N-2)}&=(\cdots,\ 0,\ S_{(N-1)(N-1)},\ 0,\ S_{(N-1)(N+1)})\\
		& \ \ \ \ \ (\cdots,\ 0,\ -h_{N-4},\ 0,\ 0)^T \\
		&=-h_{N-4}S_{(N-1)(N-1)}. \\
		R_{(N-1)(N-2)}&=(\cdots,\ -h_{N-4}',\ 0,\ h_{N-3}',\ 0)\\
		&\quad\quad (\cdots,\ 0,\ S_{(N-2)(N-2)},\ 0,\ 0,\ 0)^T \\
		&=-h_{N-4}'S_{(N-2)(N-2)}.
		\end{split}
		\end{equation*}
		Hence, we have
		\begin{equation}\label{eq57}
		h_{N-4}S_{(N-1)(N-1)}=h_{N-4}'S_{(N-4)(N-4)}.
		\end{equation}
		Since $S_{(N-1)(N-1)}^2=S_{(N-2)(N-2)}^2=1$, (\ref{eq57}) is equivalent to $h_{N-4}S_{(N-2)(N-2)}=h_{N-4}'S_{(N-1)(N-1)}$, which can be substituted into (\ref{eq55}) and lead to $S_{(N-2)N}=0$. Now the only non-zero element in $S_{\sigma N}$ is $S_{NN}$ and its magnitude is still undetermined.
		
		The equations $$L_{\sigma N}=R_{\sigma N}$$ and $$h_{N-3}S_{\sigma (N-1)}-h_{N-2}S_{\sigma (N+1)}= S_{NN}\tilde{A}_{\sigma N}'$$ yield that
		\begin{equation}
		\!\!\!\!\!\!\left(\!\!\!
		\begin{array}{c}
		\vdots\\
		-h_{N-2}S_{(N-2)(N+1)}\\
		h_{N-3}S_{(N-1)(N-1)}-h_{N-2}S_{(N-1)(N+1)}\\
		-h_{N-2}S_{N(N+1)}\\
		-h_{N-2}S_{(N+1)(N+1)}\\
		\end{array}
		\!\!\!\right)\!\!\!=\!\!\!\left(\!\!
		\begin{array}{c}
		\vdots\\
		0\\
		S_{NN}h_{N-3}'\\
		0\\
		-S_{NN}h_{N-2}'\\
		\end{array}
		\!\!\right).
		\end{equation}
		We thus have
		\begin{equation}\label{eq59}
		h_{N-3}S_{(N-1)(N-1)}-h_{N-2}S_{(N-1)(N+1)}=S_{NN}h_{N-3}',
		\end{equation}
		\begin{equation}\label{eq60}
		h_{N-2}S_{(N+1)(N+1)}=h_{N-2}'S_{NN},
		\end{equation}
		and $S_{i}=0$ for $i=1,\ \cdots,\ N-2,\ N$. (The case of $N<3$ can be easily analyzed and we omit the analysis.) Since $S$ is non-singular, we must have $S_{NN}\neq0$ and $S_{(N+1)(N+1)}\neq0$.
		
		From $L_{1(N+1)}=R_{1(N+1)}$, we have $$\tilde{A}_{1(N+1)}=\tilde{A}_{1(N+1)}'S_{(N+1)(N+1)},$$ which using (\ref{eq43}) is
		\begin{equation}
		\frac{h_{N-1}\cdot h_\beta h_2\cdots h_{N-1}}{h_\alpha h_1\cdots h_{N-2}} =\frac{h_{N-1}'\cdot h'_\beta h_2'\cdots h_{N-1}'}{h'_\alpha h_1'\cdots h_{N-2}'}S_{(N+1)(N+1)}.
		\end{equation}
		Since $|h_{i}|=|h_{i}'|$ for $1\leq i\leq N-4$, we further have
		\begin{equation}\label{eq64}
		\left|\frac{h_{N-3}h_{N-1}^2}{h_{N-2}} \right|=\left|\frac{h_{N-3}'h_{N-1}'^2}{h_{N-2}'}S_{(N+1)(N+1)}\right|.
		\end{equation}
		We substitute (\ref{eq56}) and (\ref{eq60}) into (\ref{eq64}) to obtain
		\begin{equation}\label{eq65}
		S_{NN}^2=\frac{h_{N-1}^2}{h_{N-1}'^2}.
		\end{equation}
		Now consider $L_{(N-2)(N+1)}=R_{(N-2)(N+1)}$,
		\begin{equation}
		\begin{split}
		&S_{(N-2)(N-2)}\tilde{A}_{(N-2)(N+1)} \\
		&=h_{N-4}'S_{(N-1)(N+1)}+\tilde{A}_{(N-2)(N+1)}'S_{(N+1)(N+1)}.
		\end{split}
		\end{equation}
		which is
		\begin{equation}\label{eq66}
		\begin{split}
		&h_{N-4}'S_{(N-1)(N+1)}\\
		&=\frac{h_{N-3}h_{N-1}^2}{h_{N-4}h_{N-2}}S_{(N-2)(N-2)} -\frac{h_{N-3}'h_{N-1}'^2}{h_{N-4}'h_{N-2}'}S_{(N+1)(N+1)}.
		\end{split}
		\end{equation}
		Using (\ref{eq56}), (\ref{eq57}), (\ref{eq60}) and (\ref{eq65}), we can eliminate all $h_i'$ for $i=N-4,\ \cdots,\ (N-1)$:
		\begin{equation}
		\begin{array}{rl}
		&\ \ \ \ \frac{h_{N-3}'h_{N-1}'^2}{h_{N-4}'h_{N-2}'}S_{nn}\\
		&\!\!\!\!=\frac{h_{N-3}S_{NN}}{S_{(N-1)(N-1)}}\frac{h_{N-1}^2}{S_{NN}^2} \frac{S_{(N-2)(N-2)}}{h_{N-4}S_{(N-1)(N-1)}}\frac{S_{NN}}{h_{N}S_{(N+1)(N+1)}}S_{(N+1)(N+1)}\\
		&\!\!\!\!=\frac{h_{N-3}h_{N-1}^2}{h_{N-4}h_{N-2}}S_{(N-2)(N-2)},\\
		\end{array}
		\end{equation}
		which implies that the RHS of (\ref{eq66}) is zero. Hence, $S_{(N-1)(N+1)}=0$, and $S$ is diagonal.
		
		Eq. (\ref{eq59}) is now $h_{N-3}S_{(N-1)(N-1)}=S_{NN}h_{N-3}'$, which is equivalent to
		\begin{equation}\label{eq67}
		h_{N-3}/S_{NN}=S_{(N-1)(N-1)}h_{N-3}'.
		\end{equation}
		Substituting (\ref{eq67}) into (\ref{eq56}), we obtain $S_{NN}=\pm1$. Hence, from (\ref{eq67}) and (\ref{eq65}) we know $|h_{N-3}|=|h_{N-3}'|$ and $|h_{N-1}|=|h_{N-1}'|$. Now the only undetermined unknown parameter is $h_{N}$.
		
		From $|L_{N(N+1)}|=|R_{N(N+1)}|$, we have $|\tilde{A}_{N(N+1)}|=\\|\tilde{A}_{N(N+1)}'S_{(N+1)(N+1)}|$ which using (\ref{eq43}) is
		\begin{equation}\label{eq68}
		\left|h_{N-2}+\frac{h_{N-1}^2}{h_{N-2}}\right|=\left|\left(h_{N-2}'+\frac{h_{N-1}'^2}{h_{N-2}'}\right)S_{(N+1)(N+1)}\right|.
		\end{equation}
		Substitute (\ref{eq60}) into (\ref{eq68}) to eliminate $S_{(N+1)(N+1)}$:
		\begin{equation}
		\left|h_{N-2}+\frac{h_{N-1}^2}{h_{N-2}}\right|=\left|\left(h_{N-2}'+\frac{h_{N-1}'^2}{h_{N-2}'}\right)\frac{h_{N-2}'S_{NN}}{h_{N-2}}\right|,
		\end{equation}
		which is
		\begin{equation}\label{eq70}
		\left|h_{N-2}^2+h_{N-1}^2\right|=\left|h_{N-2}'^2+h_{N-1}'^2\right|.
		\end{equation}
		Eq. (\ref{eq70}) means $|h_{N-2}|=|h_{N-2}'|$. Together with Lemma \ref{lem5}, we prove that all of the unknown parameters are able to be estimated by the two-qubit sensor when $N$ is odd.
We conclude that when measuring $Z_\alpha Y_\beta$, the information extracted from the sensor is enough to identify the coupling parameters. Therefore, the answer to Problem \ref{problem1} is positive.
\end{proof}
\section{Capability test when  measuring $Y_\alpha Z_\beta$}\label{section4}
For \textbf{Problem 2}, where the initial state is chosen to be $X_\beta$ and the measurement is $Y_\alpha Z_\beta$, the accessible set is
\begin{equation}\label{G3y1z2}
G^3=\{Y_\alpha Z_\beta,\ X_\beta,\ Z_\beta Y_1,\ Y_\alpha X_\beta Y_1,\ X_\alpha Y_\beta Y_1,\ Y_\alpha Z_1,\ \cdots \ \}.
\end{equation}
Denote the total number of spins in the spin chain system as $N$, the number of operators in $G^3$ is $\frac{1}{2}(N+2)^3-\frac{1}{2}(N+2)^2$ \cite{qiyu2019accessibleset}. Therefore, the dimension of the state space model increases quickly with increasing $N$.

The generation of operators in $G^3$ follows a special pattern (see Appendix \ref{ASY1Z2}). Given the accessible set $G^3$, we choose the order of element operators in the state vector $\mathbf{x}$ to be the order as shown in Fig. \ref{GenerationRulesN3} and Fig. \ref{GenerationRulesN4} which follows a set of specific generation rules \cite{qiyu2019accessibleset}. As an example, we present the matrix $A$ for $N=2$ in \eqref{AmatrixY1Z2}. For higher dimensional cases, we may employ a similar method to obtain the matrix $A$ , although this task is usually complicated. Since the measurement is $Y_\alpha Z_\beta$ and the initial state is $X_\beta$, we have
\renewcommand{\arraystretch}{0.0}
{\tiny
\begin{figure*}
\begin{equation}\label{AmatrixY1Z2}
A = \setlength\arraycolsep{0.0pt}
 \left(
\begin{array}{cccccccccccccccccccccccc}
0           &-h_\alpha &0             &-h_\beta &0            &h_\beta &0           &0            &0           &0           &0            &        0 &        0 &0           &0         &0         & 0         &0         & 0         &0         &0         &         0&         0& 0\\
h_\alpha&0             &h_\beta    &0            &0            &0          &0           &0            &0            &0          &0            &        0 &        0 &0           &0         &0         & 0         &0         & 0         &0         &0         & 0        & 0        & 0\\
0           &-h_\beta  &0             &-h_\alpha&0            &0          &h_\alpha&0            &0            &-h_1      &0            &        0 &        0 &0           &0         &0         & 0         &0         & 0         &0         &0         & 0        & 0        & 0\\
h_\beta &0             &h_\alpha  &0             &-h_\beta &0          &0           &-h_\alpha&0            &0          &-h_1        &        0 &        0 &0           &0         &0         & 0         &0         & 0         &0         &0         & 0        & 0        & 0\\
0          & 0            & 0            &h_\beta    &0           &-h_\beta&0           &0            &-h_\alpha&0          &0            &h_1       &      0  &-h_1       &0         &0         & 0         &0         & 0         &0         &0         & 0        & 0        & 0\\
-h_\beta& 0            & 0            &0             &h_\beta &0           &0           &0            &0            &0          &0            &0          &      0  &0            &h_1      &0         & 0         &0         & 0         &0         &0         & 0        & 0        & 0\\

0          & 0            & -h_\alpha &0             &0           &0           &0           &0            &h_\alpha &0          &0            &0          &      0  &0            &0         &-h_1    & 0         &0         & 0         &0         &0         & 0        & 0        & 0\\
0          & 0            & 0           &h_\alpha   &0           &0          &-h_\alpha&0            &-h_\beta  &0          &0            &0          &      0  &0            &0         &0         &-h_1     &0         & 0         &0         &0         & 0        & 0        & 0\\
0          & 0            & 0            &0             &h_\alpha&0           &0           &h_\beta   &0            &0          &0            &0          &      0  &0            &0         &0         & 0         &h_1       & 0         &-h_1   &0         & 0        & 0        & 0\\
0          & 0            & h_1         &0             &0           &0           &0           &0            &0            &0         &-h_\alpha&0          &      0  &0            &0        &h_\alpha& 0         &0         & 0         &0         &0         & 0        & 0        & 0\\
0          & 0            & 0            &h_1         &0           &0           &0           &0            &0            &h_\alpha&0            &h_\beta &      0  &0            &0        &0        &-h_\alpha&0         & 0         &0         &0         & 0        & 0        & 0\\
0          & 0            & 0            &0             &-h_1     &0           &0           &0            &0            &0          &-h_\beta  &0          &h_1     &0            &0         &0         & 0       &-h_\alpha& 0         &0         &0         & 0        & 0        & 0\\

0          & 0            & 0            &0             &0           &0           &0           &0            &0            &0          &0            &-h_1    &      0  &h_1        &0         &0         & 0         &0       &-h_\alpha&0         &0         & 0        & 0        & 0\\
0          & 0            & 0            &0             &h_1        &0           &0           &0            &0            &0          &0            &0          &-h_1   &0            &h_\beta&0         & 0         &0         & 0      &-h_\alpha&0         & 0        & 0        & 0\\
0          & 0            & 0            &0             &0           &-h_1      &0           &0            &0            &0          &0            &0          &      0  &-h_\beta &0         &0         & 0         &0         & 0         &0         &0         & 0        & 0        & 0\\
0          & 0            & 0            &0             &0           &0           &h_1        &0            &0          &-h_\alpha&0            &0          &      0  &0            &0         &0         &h_\alpha&0         & 0         &0       &-h_\beta& 0        & 0        & 0\\
0          & 0            & 0            &0             &0           &0           &0           &h_1         &0            &0          &h_\alpha &0          &      0  &0            &0        &-h_\alpha& 0       &h_\beta&0         &0         &0        &-h_\beta& 0        & 0\\
0          & 0            & 0            &0             &0           &0           &0           &0            &-h_1        &0          &0           &h_\alpha&      0  &0            &0         &0         &-h_\beta&0         &h_1      &0         &0         & 0        &h_\beta& 0\\

0          & 0            & 0            &0             &0           &0           &0           &0            &0            &0          &0            &0        &h_\alpha&0            &0         &0         & 0         &-h_1   & 0         &h_1     &0         & 0       & 0   &-h_\beta\\
0          & 0            & 0            &0             &0           &0           &0           &0            &h_1        &0          &0            &0          &      0  &h_\alpha &0         &0         & 0         &0         &-h_1       &0         &0         & 0        & 0        & 0\\
0          & 0            & 0            &0             &0           &0           &0           &0            &0            &0          &0            &0          &      0  &0            &0         &h_\beta& 0         &0         & 0         &0         &0       &h_\alpha& 0        & 0\\
0          & 0            & 0            &0             &0           &0           &0           &0            &0            &0          &0            &0          &      0  &0            &0         &0         &h_\beta&0         & 0         &0         &-h_\alpha& 0       &-h_\beta& 0\\
0          & 0            & 0            &0             &0           &0           &0           &0            &0            &0          &0            &0          &      0  &0            &0         &0         & 0        &-h_\beta& 0         &0         &0         &h_\beta& 0        &h_1\\
0          & 0            & 0            &0             &0           &0           &0           &0            &0            &0          &0            &0          &      0  &0            &0         &0         & 0         &0         &-h_\beta&0         &0         & 0        &-h_1    & 0\\
\end{array}
\right)
\end{equation}
\end{figure*}}
 \begin{equation}\label{BY1Z2}
 B=(0 \  1\  0 \  0\  0\  \cdots\ )^T
 \end{equation}
 and
 \begin{equation}\label{CY1Z2}
C=(1 \  0\  0\  0 \  0 \ \cdots\ ).
 \end{equation}
Numerical results show that for an arbitrary $N$, the state space model obtained for this case may not be minimal. To use the STA method, one has to obtain a minimal state space model with the assistance of the SPT method \cite{yuanlong2020}. However, the dimension of the state space model increases rapidly with the order of $O(N^3)$ while the corresponding minimal system has a much lower dimension. Thus, even with the help of the SPT method, the structure of the state space model will be dramatically changed when obtaining the minimal system, which makes the STA method difficult to be straightforwardly applied.

Here, we leave the sensing capability determination problem for an arbitrary $N$ as an open problem and provide readers with the Gr$\ddot{\text{o}}$bner basis method in \cite{Akira} to demonstrate how to determine the capability for a given $N$ using this method. Gr$\ddot{\text{o}}$bner basis method is a widely used algorithmic solution for a set of multivariate polynomials \cite{BB2001}.

The main idea of Gr$\ddot{\text{o}}$bner basis method is that: for a given state space model, one can obtain the system transfer function $G(s)$ with unknown parameters according to \eqref{StateSpace2Transfer}. On the other hand, one can also reconstruct a system transfer function $\hat{G}(s)$ from the measurement data using the ERA method \cite{junzhang2014}. Since the two transfer functions describe the same input-output behavior, we have
\begin{equation}\label{Transfequal}
 \hat{G}(s)=G(s).
\end{equation}
By equating the coefficients in \eqref{Transfequal}, a polynomial set $F$ with respect to unknown parameters in $\{h_i\}$ can be obtained. The unknown parameters $\{h_i\}$ can then be inferred by solving the polynomials in the set $F$. The sensing capability is then related to analysing the number of the solutions of $F$. If a solution of $F$ exists and is unique, the parameters are identifiable. Details of the algorithm for Hamiltonian identifiability can be found in \cite{Akira}.



As an example, we illustrate the use of the Gr$\ddot{\text{o}}$bner basis method for determining the capability of quantum sensor when there are two spins in the object system.
We have the following conclusion.
\begin{theorem}\label{Y1Z2identifiable}
	For a system with the state space model given as \eqref{AmatrixY1Z2}, \eqref{BY1Z2} and \eqref{CY1Z2} and , the unknown parameters $\{h_i\}$ appearing in matrix $A$ can be identified.
\end{theorem}
\begin{proof}
The Hamiltonian is
\begin{equation}
\begin{split}
H=&\frac{h_\alpha}{2}(X_\alpha X_\beta+Y_\alpha Y_\beta)+\frac{h_\beta}{2}(X_\beta X_1+Y_\beta Y_1)\\
&+\frac{h_1}{2}(X_1 X_2+Y_1 Y_2).
\end{split}
\end{equation}
The corresponding accessible set is
\begin{equation}\label{G^3_4}
\begin{split}
G^3_4=   \{ &Y_\alpha Z_\beta,\ X_\beta,\ Z_\beta Y_1,\ Y_\alpha X_\beta Y_1,\ Y_\alpha Z_1,\   Y_\alpha Y_\beta X_1,\ X_\alpha Y_\beta Y_1,\\
&Z_\alpha Y_1,\ Z_\alpha X_\beta Z_1,\  Y_\alpha Y_\beta Z_1Y_2,\ Y_\alpha X_1Y_2,\ Z_\alpha X_\beta X_1Y_2,\\
&Y_\alpha Z_2,\  Z_\alpha X_\beta Z_2,\ Z_\alpha Z_\beta Y_1Z_2, \ Y_\alpha Y_1X_2,\ Z_\alpha X_\beta Y_1X_2,\\
& Z_\alpha Z_\beta X_2,\ Y_\alpha X_\beta Z_1X_2,\ Z_\alpha Z_1X_2,\ Z_\alpha Y_\beta X_1X_2,\\
&Z_\beta Z_1X_2,\ X_\alpha Y_\beta Z_1X_2,\ X_\alpha X_1X_2\}.
\end{split}
\end{equation}
The number of element operators in $G^3_2$ is 24. The corresponding matrix $A$  is given in \eqref{AmatrixY1Z2}. $B$ and $C$ are given in \eqref{BY1Z2} and \eqref{CY1Z2}. We have
\begin{equation}
\begin{split}
&G(s)=C(sI-A)^{-1}B \\
&=-\frac{h_\alpha s^{10}+(10h_\alpha^3+7h_\alpha h_\beta^2+11h_\alpha h_1^2)s^8+\ \cdots}{s^{12}+11(h_\alpha^2+h_\beta^2+h_1^2)s^{10}+\ \cdots}.
\end{split}
\end{equation}
We sort the numerator and the denominator of the transfer function by the order of $s$ and only present items of the highest and second highest orders. The residual terms of lower order are omitted since the items of highest and second higher order are sufficient for determining the sensor capability. Moreover, we assume that the transfer function obtained by using the measurement data is
\begin{equation}
\hat{G}(s)=-\frac{v_1 s^{10}+v_2s^8+\ \cdots}{s^{12}+v_3s^{10}+\ \cdots}.
\end{equation}
From \eqref{Transfequal}, we have the following equation
\begin{equation}\label{N4Transfer}
\begin{split}
&-\frac{h_\alpha s^{10}+(10h_\alpha^3+7h_\alpha h_\beta^2+11h_\alpha h_1^2)s^8+\ \cdots}{s^{12}+11(h_\alpha^2+h_\beta^2+h_1^2)s^{10}+\ \cdots}\\
&=-\frac{v_1 s^{10}+v_2s^8+\ \cdots}{s^{12}+v_3s^{10}+\ \cdots},
\end{split}
\end{equation}
where $v_1$, $v_2$, $v_3,\  \cdots$ are real values obtained from the ERA method \cite{junzhang2014}.
Let $\theta_1=h_\alpha, \theta_2=h_\beta^2, \theta_3=h_1^2$, the polynomial equations of $\theta_1, \theta_2$ and $\theta_3$ are:
\begin{equation}
\begin{split}
 \theta_1&=v_1, \\
 10\theta_1^3+7\theta_1\theta_2+11\theta_1\theta_3&=v_2 ,\\
 11(\theta_1^2+\theta_2+\theta_3)&=v_3 , \\
\end{split}
\end{equation}
The Gr$\ddot{\text{o}}$bner basis of the above polynomials takes the following form:
\begin{equation}
\mathcal{G}= \{ \theta_1 - a_1,\  \theta_2 - a_2,\   \theta_3 - a_3\},
\end{equation}
where
\begin{equation}
\begin{split}
a_1&=v_1;\\
a_2&=\frac{-v_1^3+v_1 v_3-v_2}{4 v_1};\\
a_3&=\frac{-33 v_1^3-7 v_1 v_3+11 v_2}{44 v_1}.
\end{split}
\end{equation}
It can be seen that there is only one solution for the magnitudes of all of the parameters. Therefore, the sensor is capable of identifying these parameters.
\end{proof}



\section{Conclusion}\label{section5}
We have investigated the capability of a class of qubit sensors. The object system is a spin chain system whose Hamiltonian is described by unknown parameters. Qubit sensors are coupled to the object system in order to achieve information extraction. By initializing and probing the qubit sensor, we aim to estimate all of the unknown parameters. With a restricted initial setting and measurement schemes, we find that one qubit sensor is not capable of fully performing the estimation task. To solve this problem, we focus on analyzing two-qubit sensors. To determine their capability, the STA method and the Gr$\ddot{\text{o}}$bner basis method are employed. We prove that a two-qubit sensor can estimate all of the unknown parameters, which reveals an effective way to improve the capability of quantum sensors through increasing the number of qubits in the sensor.
\appendix
\begin{ack}                               
	The authors would like to thank Matthew James, Akira Sone and Paola Cappellaro for constructive suggestions and helpful discussion.
\end{ack}

\section{Accessible set}\label{APPX1}
\subsection{Accessible set $G^1$ when  measuring $Z_\alpha Y_\beta$}\label{ASZ1Y2}
Given the initial measurement operator $Z_\alpha Y_\beta$, the accessible set can expand itself using the iterative rules given in \cite{qiyu2019accessibleset}. We have
\begin{equation}
G^1=\{O_1^1,\ O_2^1,\  \cdots,\  O_i^1,\  \cdots\  \}.
\end{equation}
where
\begin{equation}
O_i^1=
\begin{cases}
Z_{\alpha}Z_{\beta} \cdots Z_{i-1}Y_i, \quad \text{$i$ is even},\\
Z_{\alpha}Z_{\beta} \cdots Z_{i-1}X_i, \quad \text{$i$ is odd} .
\end{cases}
\end{equation}

\subsection{Accessible set $G^2$ when measuring $Y_\alpha Z_\beta$}\label{ASY1Z2}
For measurement scheme $Y_\alpha Z_\beta$, the generation of accessible set does have a pattern when there are an arbitrary number of spins in the object system. However, it is difficult to give a simplified uniform description for the general pattern. Here, we provide accessible sets for the cases in which there are only $3$, $4$ and $5$ spins in the whole system (including the sensor). Let the accessible set of $i$ spins be $G^3_i$ and $G^3_{\lfloor i}=G^3_i-G^3_{ i-1}$. The operator $``-"$ acting on arbitrary sets $A$ and $B$ by $A-B$ indicates the subtraction of set $B$ from set $A$. When the measurement is $Y_\alpha Z_\beta$, we have $G^3_3\subset G^3_4\subset G^3_5$. We present the set $G^3_3$ in Fig. \ref{GenerationRulesN3}, the set $G^3_{\lfloor 4}$ in Fig. \ref{GenerationRulesN4} and the set $G^3_{\lfloor 5}$ in Fig. \ref{GenerationRulesN5}. It can be seen that the generation of set $G^3_{\lfloor i}$ has a shared pattern. 
\begin{figure}	
	\centering		
	\includegraphics[width=8cm]{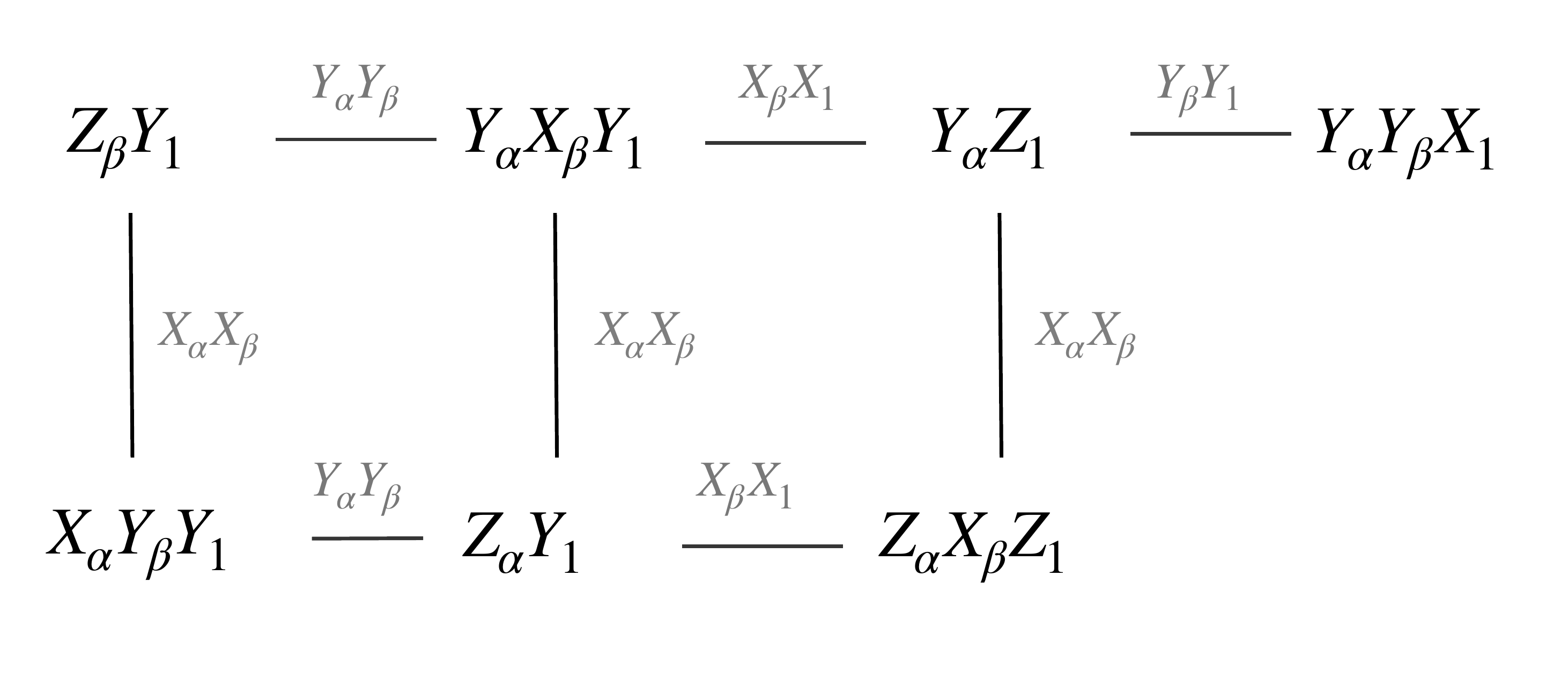}		
	\caption{The generation procedure for the system with $3$ spins (including the sensor). The operators on the vertices form the set $G^3_3$.}
	\label{GenerationRulesN3}		
\end{figure}
\begin{figure}	
	\centering		
	\includegraphics[width=8cm]{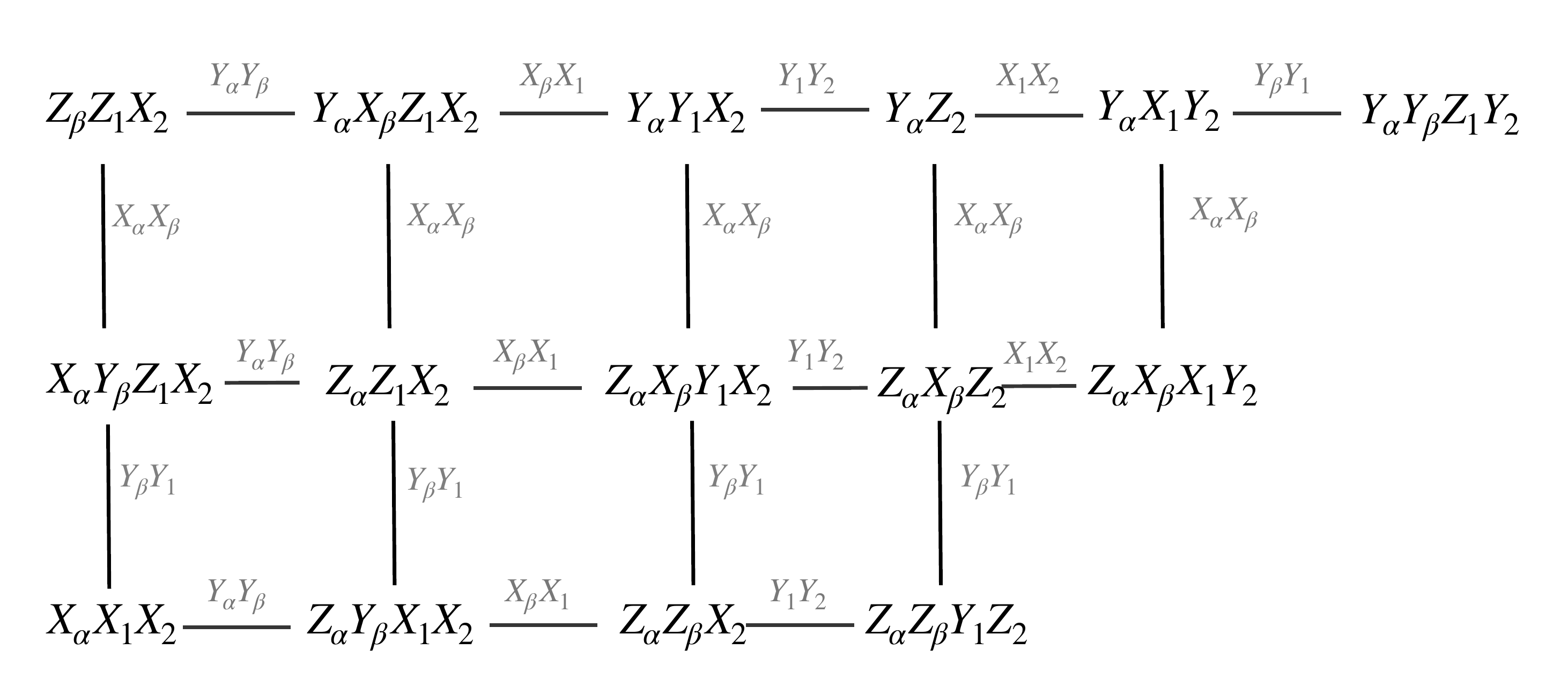}		
	\caption{The generation procedure for $G_{\lfloor 4}^3$. The bold operators on the vertices form the set $G_{\lfloor 4}^3$. }
	\label{GenerationRulesN4}		
\end{figure}
\begin{figure}	
	\centering		
	\includegraphics[width=8cm]{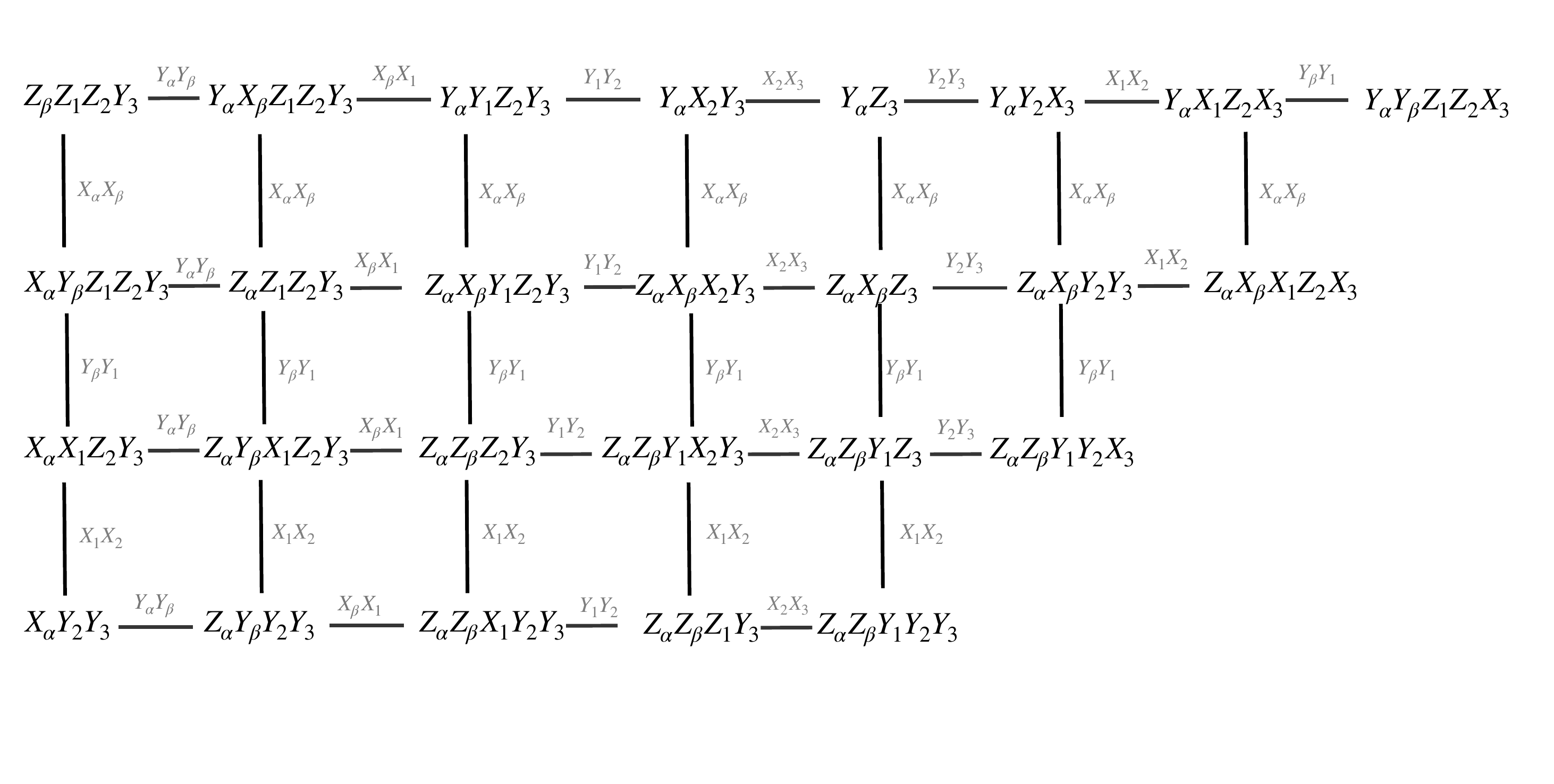}		
	\caption{The generation procedure for $G^3_{\lfloor 5}$. The bold operators above form the set $G^3_{\lfloor 5}$.}
	\label{GenerationRulesN5}		
\end{figure}
\section{PROOF OF LEMMA \ref{lem1}}\label{APPlem1}
By induction we can prove $$A^kB=[(*,\ \cdots,\ *,\ (-1)^{k}h_\alpha h_\beta \prod_{i=1}^{k-2}h_i,\ 0,\ \cdots,\ 0)^T]_{(N+2)\times1}$$ for $1\leq k\leq (N+1)$ where $*$ are polynomials on $h_\alpha$, $h_\beta$ and $\{h_i\}$, and the last $N+1-k$ elements are zero. Therefore, $\text{CM}$ is an upper triangular matrix and its determinant is $$\det(\text{CM})=h_\alpha^{N+1}h_\beta^N\prod_{k=3}^{N+1} (-1)^{k}\prod_{i=1}^{k-2}h_i,$$ which is non-zero for almost any value of $\{h_i\}$. Hence $\text{CM}$ is almost always full-ranked.

\section{PROOF OF LEMMA \ref{lem2}}\label{APPlem2}
\setlength\arraycolsep{1.5pt}\def\arraystretch{1.5}
Since the rank of the observability matrix does not change when the system undergoes a similarity transformation, we can consider $\overline{\text{OM}}$ of the system $(\bar{A},\bar{C})$, where $\bar{A}=PAP^{-1}$ and $\bar{C}=CP^{-1}$ for any non-singular $P$. We take $I_{(N+2)\times(N+2)}$ and rearrange all its even rows in the ascending order, followed by all of the odd rows also in the ascending order, to form $P$. After this similarity transformation, $\bar{A}=\left(\begin{matrix}0& \tilde{A} \\ -\tilde{A}^T &0\end{matrix}\right)$ and $\bar{C}=(1,0,\cdots,0)$, where
$$\tilde{A}=\setlength\arraycolsep{2pt}\def\arraystretch{2} \left(\begin{matrix} -h_\alpha & h_\beta&0 & 0&\cdots \\ 0 & -h_1 &h_2 & 0& \cdots \\ 0 &0 & -h_2& h_3 & \cdots \\ \vdots& \vdots &\vdots &\ddots &\ddots \end{matrix}\right).$$
Partition $\bar{C}$ as $\bar{C}=(\tilde{C}\ \ 0_{1\times\frac{N+2}{2}})$ and we have
\begin{equation}
\overline{\text{OM}}=\left(
\begin{array}{cc}
\tilde{C}&0\\
0&\tilde{C}\tilde{A}\\
-\tilde{C}\tilde{A}\tilde{A}^T&0\\
0&-\tilde{C}\tilde{A}\tilde{A}^T\tilde{A}\\
\cdots&\cdots\\
\tilde{C}(-\tilde{A}\tilde{A}^T)^\frac{N}{2}&0\\
0&\tilde{C}(-\tilde{A}\tilde{A}^T)^\frac{N}{2}\tilde{A}\\
\end{array}\nonumber\right).
\end{equation}
Hence, it suffices to prove that
$$Q=( \tilde{C},\  -\tilde{C}\tilde{A}\tilde{A}^T, \ \cdots\  \tilde{C}(-\tilde{A}\tilde{A}^T)^\frac{N}{2})^T$$ is non-singular. In fact, $Q$ is lower triangular with the diagonal line $(1,\ h_\beta h_1,\ h_\beta h_1 h_2  h_3,\ \cdots,\ h_\beta \prod_{i=1}^{(N-1)} h_i)$. Therefore $Q$ is non-singular for almost any value of $\{h_i\}$, which indicates $\text{OM}$ is almost always of full rank.

\section{PROOF OF LEMMA \ref{lem3}}\label{APPlem3}
\setlength\arraycolsep{1.5pt}\def\arraystretch{1.5}
Using the PBH test \cite{hespanha 2009}, a system $(A,C)$ is observable if
\begin{equation}\label{eqc1}
\left(\begin{matrix}
C\\
A-\lambda I\\
\end{matrix}\right)
\end{equation}
is always column-full-rank for any complex number $\lambda$. Specifically, we take $\lambda=0$ in (\ref{eqc1}) and prove that the matrix
\begin{equation}\label{eqc2}
\left(\begin{matrix}
C\\
A-0 I\\
\end{matrix}\right)=
\left(
\begin{array}{cccc}
0 & 1 & 0 & \cdots \\
0 & h_\alpha & 0 & \cdots \\
-h_\alpha & 0 & h_\beta & \cdots \\
 0 & -h_\beta &  0& \cdots \\
 \vdots & \vdots & \vdots & \ddots \\
\end{array}\right)
\end{equation}
is rank deficient in column. Since the former two rows are all zeros except the elements in the second column in (\ref{eqc2}), we can remove its former two rows and the second column and consider the remaining sub-matrix
\begin{equation}\label{eqc3}
\left(
\begin{array}{ccccc}
-h_\alpha & h_\beta &0  &0 & \cdots\\
0 & 0 & h_1 & 0&\cdots \\
0& -h_1 &0 & h_2& \cdots\\
\vdots & \vdots & \vdots & \vdots &\ddots\\
\end{array}\right).
\end{equation}
We only need to prove that (\ref{eqc3}) has determinant zero. Denote $Z$ as the matrix obtained by deleting the first row and the first column of (\ref{eqc3}). Then $Z$ is an odd-dimensional antisymmetric matrix, which is always singular because $-\det(Z)=\det(-Z)=\det(Z^T)=\det(Z)$. Hence, the system $(A,C)$ is always unobservable.
\section{PROOF OF LEMMA \ref{lem4}}\label{APPlem4}
\setlength\arraycolsep{1.5pt}\def\arraystretch{1.5}
We prove the conclusions for $N\geq9$. The cases of $N<9$ can be verified straightforwardly.

\setlength\arraycolsep{1.5pt}\def\arraystretch{1.5}
It is clear that $[\bar{P}\ \ \mathbf{p}]$ has the following structure:
\begin{equation}\label{eqd1}
\setlength\arraycolsep{2.4pt}
\left(\begin{array}{ccccccc}
0 & 1 & 0 & 0 & 0 &\cdots & 0\\
-h_\alpha & 0 & h_\beta & 0 & 0 & \cdots & 0\\
0 & * & 0 & h_\beta h_1 & 0 & \cdots & 0\\
* & 0 & * & 0 & h_\beta h_1 h_2 & \cdots & 0\\
\hdotsfor{7} \\
* & 0 & * & 0 & * & \cdots  & h_\beta\prod_{i=1}^{N-1}h_i \\
\end{array}\right)_{(N+1)\times(N+2)},
\end{equation}
where the $*$s are polynomials in $h_i$. Hence, $$\mathbf{p}=(0,\ \cdots,\ 0,\ h_\beta\prod_{i=1}^{N-1}h_i)^T.$$

From $\bar{P}_{1\sigma}\cdot (\bar{P}^{-1})_{\sigma (N+1)}=0$, we know $(\bar{P}^{-1})_{2(N+1)}=0$. Since $(\bar{P}^{-1})_{ij}=[\text{adj}(\bar{P})]_{ij}=(-1)^{i+j}M(\bar{P})_{ji}$, $M(\bar{P})_{(N+1)2}=0$. Then from $\bar{P}_{3\sigma}\cdot (\bar{P}^{-1})_{\sigma (N+1)}=0$, we know $(\bar{P}^{-1})_{4(N+1)}=0$ and $M(\bar{P})_{(N+1)4}=0$. Continuing this procedure, we finally know $(\bar{P}^{-1})_{kn}=M(\bar{P})_{(N+1)k}=0$ for every even $k\leq (N+1)$.


We have $(\bar{P}^{-1})_{1(N+1)}\det(\bar{P})=  \text{adj}(\bar{P})_{1(N+1)}=\\(-1)^{N+2}M(\bar{P})_{(N+1)1}=-M(\bar{P})_{(N+1)1}$. From (\ref{eqd1}), $\bar{P}$ becomes lower triangular after deleting $\bar{P}_{(N+1)\sigma}$ and $\bar{P}_{\sigma1}$. Therefore, $-K=(\bar{P}^{-1})_{1(N+1)}\det(\bar{P})=-M(\bar{P})_{(N+1)1}=-\prod_{j=2}^{N}\prod_{i=2}^j h_i=-\prod_{i=2}^{N} h_i^{(N-1)-i}$.


Now we calculate $M(\bar{P})_{(N+1)k}$ with odd $k\geq1$ through induction on $k$. For $k=1$, $\bar{P}_{2\sigma}\cdot (\bar{P}^{-1})_{\sigma (N+1)}=0=-h_\alpha M(\bar{P})_{(N+1)1}+h_\beta M(\bar{P})_{(N+1)3}$, which implies $M(\bar{P})_{(N+1)3}\\=M(\bar{P})_{(N+1)1}h_\alpha /h_\beta$. Using similar procedures, it can be verified that $M(\bar{P})_{(N+1)k}=M(\bar{P})_{(N+1)1}(h_\alpha\prod_{i=1}^{(k-3)/2} h_{2i-1})/ \\(h_\beta\prod_{i=1}^{(k-3)/2} h_{2i})$ holds for $k=5,7$. Suppose it holds for all odd $k\in [7,m]$, where $m\in[7,N-2]$ is also odd. Let the $(m-1)$th row of $[\bar{P}\ \ \mathbf{p}]$ be
\begin{equation}\label{eqd2}
(t_1,\ 0,\ t_3,\ 0,\ \cdots,\ 0,\ t_{m-2},\ 0,\ h_\beta\prod_{i=1}^{m-3} h_i,\ 0,\ 0).
\end{equation}
Then the $m$th row of $[\bar{P}\ \ \mathbf{p}]$ is (\ref{eqd2}) multiplied by $A$ from the right:
\begin{equation}\label{eqd3}
\begin{array}{rl}
&(0,\ h_\alpha t_1-h_\beta t_3,\ 0,\ h_1t_3-h_2t_5,\ 0, \ \cdots,\\
&\ \ \ \ 0,\ h_{m-4}t_{m-2}-h_{m-3}h_\beta\prod_{i=1}^{m-3}h_i,\ 0,\ h_\beta\prod_{i=1}^{m-2}h_i,\ 0).
\end{array}
\end{equation}
and the $(m+1)$th row of $[\bar{P}\ \ \mathbf{p}]$ is (\ref{eqd3}) multiplied by $A$ from the right:
\begin{equation}\label{eqd4}
\begin{array}{rl}
&\!\!\!\!(-h_\alpha(h_\alpha t_1-h_\beta t_3),\ 0,\ h_\beta(h_\alpha t_1-h_\beta t_3)-h_1(h_1t_3-h_2t_5),\ 0,\\
&\ \cdots,\  0,\ h_{m-3}(h_{m-4}t_{m-2}-h_{m-3}h_\beta\prod_{i=1}^{m-3}h_i)-\\
&h_{m-2}h_\beta\prod_{i=1}^{m-2}h_i,\ 0,\ h_\beta\prod_{i=1}^{m-1}h_i).
\end{array}
\end{equation}

By a straightforward but tedious calculation, we have
\begin{equation}\label{eqd5}
\!\!\!\!\!\!\!\!\!\!\!\!\begin{array}{rl}
& \bar{P}_{(m+1)\sigma}\cdot (\bar{P}^{-1})_{\sigma (N+1)}\det\bar{P}\\
&=\frac{-Kh_{m-2}(h_\alpha\prod_{i=1}^{(m-3)/2}h_{2i-1})}{(h_\beta\prod_{i=1}^{(m-3)/2}h_{2i})(h_\beta\prod_{i=1}^{m-2}h_i)} -(h_\beta\prod_{i=1}^{m-1}h_i)\cdot M(\bar{P})_{(N+1)(m+2)}\\
&=0
\end{array}
\end{equation}
Therefore, $$M(\bar{P})_{(N+1)(m+2)}=\\
 K(h_\alpha\prod_{i=1}^{(m-1)/2}h_{2i-1})/(h_\beta\prod_{i=1}^{(m-1)/2}h_{2i}).$$


Now we consider $\bar{P}_{(N+1)\sigma}\cdot (\bar{P}^{-1})_{\sigma (N+1)}\det\bar{P}=\det\bar{P}$, which equals to the last equation in (\ref{eqd5}) with $m=N$ and $M(\bar{P})_{(N+1)(m+2)}=0$. Hence,
\begin{equation}
\begin{array}{rl}
\det{\bar{P}}&=-K h_{N-2}(h_\alpha\prod_{i=1}^{(N-3)/2}h_{2i-1})/(h_\beta\prod_{i=1}^{(N-3)/2}h_{2i})\\
&\ \ \ \cdot(h_\beta\prod_{i=1}^{N-2}h_i)\\
&=-(h_{\beta}^{N-1}\prod_{i=1}^{N-2}h_i^{(N-1)-i})h_{N-2}(h_\alpha\prod_{i=1}^{(N-3)/2}h_{2i-1})\\
 &\ \ \ \ /(h_\beta\prod_{i=1}^{(N-3)/2}h_{2i})(h_\beta\prod_{i=1}^{N-2}h_i)\\
&=h_\alpha h_{\beta}^{N-1}h_{N-2}^3\prod_{i=1}^{(N-3)/2}h_{2i-1}^{N+2-2i}h_{2i}^{(N-1)-2i}.\\
\end{array}
\end{equation}





\end{document}